\documentclass[11pt]{article}

\usepackage{amsmath}
\usepackage{amsthm}
\usepackage{graphicx} 
\usepackage{array} 
\usepackage{bbm}
\usepackage{thm-restate}

\usepackage{fancybox}

\newenvironment{fminipage}%
  {\begin{Sbox}\begin{minipage}}%
  {\end{minipage}\end{Sbox}\fbox{\TheSbox}}

\newenvironment{algbox}[0]{\vskip 0.2in
\noindent 
\begin{fminipage}{6.3in}
}{
\end{fminipage}
\vskip 0.2in
}

\DeclareMathOperator*{\poly}{poly}
\DeclareMathOperator{\diag}{diag}

\newcommand{\norm}[1]{\left\lVert#1\right\rVert}
\def\abs#1{\left|#1  \right|}

\usepackage{amsmath, amssymb, amsfonts, verbatim}
\usepackage{hyphenat,epsfig,subfigure,multirow}
\usepackage{multicol}
\usepackage{nicefrac}
\usepackage{paralist}
\usepackage{enumitem}
\usepackage{tikz}
\usetikzlibrary{positioning}

\usepackage[ruled]{algorithm2e}

\DeclareFontFamily{U}{mathx}{\hyphenchar\font45}
\DeclareFontShape{U}{mathx}{m}{n}{
      <5> <6> <7> <8> <9> <10>
      <10.95> <12> <14.4> <17.28> <20.74> <24.88>
      mathx10
      }{}
\DeclareSymbolFont{mathx}{U}{mathx}{m}{n}
\DeclareMathSymbol{\bigtimes}{1}{mathx}{"91}

\usepackage{tcolorbox}
\tcbuselibrary{skins,breakable}
\tcbset{enhanced jigsaw}

\usepackage[normalem]{ulem}
\usepackage[compact]{titlesec}

\definecolor{DarkRed}{rgb}{0.5,0.1,0.1}
\definecolor{DarkBlue}{rgb}{0.1,0.1,0.5}

\usepackage{nameref}
\definecolor{ForestGreen}{rgb}{0.1333,0.5451,0.1333}
\definecolor{Red}{rgb}{0.9,0,0}
\usepackage[linktocpage=true,
	pagebackref=true,colorlinks,
	linkcolor=DarkRed,citecolor=ForestGreen,
	bookmarks,bookmarksopen,bookmarksnumbered]
	{hyperref}
\usepackage[noabbrev,nameinlink]{cleveref}
\crefname{property}{property}{Property}
\creflabelformat{property}{(#1)#2#3}
\crefname{equation}{eq}{Eq}
\creflabelformat{equation}{(#1)#2#3}

\usepackage{bm}
\usepackage{url}
\usepackage{xspace}
\usepackage[mathscr]{euscript}

\usepackage{tikz}
\usetikzlibrary{arrows}
\usetikzlibrary{arrows.meta}
\usetikzlibrary{shapes}
\usetikzlibrary{backgrounds}
\usetikzlibrary{positioning}
\usetikzlibrary{decorations.markings}
\usetikzlibrary{patterns}
\usetikzlibrary{calc}
\usetikzlibrary{fit}
\usetikzlibrary{snakes}

\usepackage{mdframed}

\usepackage[noend]{algpseudocode}
\makeatletter
\def\BState{\State\hskip-\ALG@thistlm}
\makeatother

\usepackage{cite}
\usepackage{enumitem}

\usepackage[margin=1in]{geometry}

\newtheorem{theorem}{Theorem}[section]

\newtheorem{lemma}{Lemma}[section]

\newtheorem{corollary}[lemma]{Corollary}
\newtheorem{claim}[lemma]{Claim}

\newtheorem{definition}[lemma]{Definition}

\newtheorem*{theorem*}{Theorem}
\newtheorem*{claim*}{Claim}
\newtheorem*{proposition*}{Proposition}
\newtheorem*{lemma*}{Lemma}
\newtheorem*{problem*}{Problem}

\crefname{lemma}{Lemma}{Lemmas}
\crefname{claim}{Claim}{Claims}

\newtheoremstyle{restate}{}{}{\itshape}{}{\bfseries}{~(restated).}{.5em}{\thmnote{#3}}
\theoremstyle{restate}

\allowdisplaybreaks

\title{
Entrywise Approximate Solutions for SDDM Systems in Almost-Linear Time
}
\author{
Angelo Farfan\footnote{dstaff@mit.edu}
\and
Mehrdad Ghadiri\footnote{mehrdadg@mit.edu}
\and Junzhao Yang\footnote{junzhaoy@andrew.cmu.edu}
}
\date{}

\newcommand{\Mbegin}{\left[\begin{matrix}}
\newcommand{\Mend}{\end{matrix}\right]}

\newcommand{\callalg}[1]{\hyperref[alg:#1]{\textsc{#1}}}
\newcommand{\labelalg}[1]{\label{alg:#1}}

\newcommand{\epsL}{\eps_L}

\newcommand{\ExpSet}{\mathsf{Exp}}

\newcommand{\rin}{r_{\text{in}}}
\newcommand{\rout}{r_{\text{out}}}

\newcommand{\Cores}{\text{Cores}}
\newcommand{\cnt}{\mathsf{cnt}}

\newcommand\ee{\boldsymbol{\mathit{e}}}
\newcommand\xx{\boldsymbol{\mathit{x}}}
\newcommand\xxtil{\widetilde{\boldsymbol{\mathit{x}}}}
\newcommand\xxbar{\overline{\boldsymbol{\mathit{x}}}}
\newcommand\xxhat{\widehat{\boldsymbol{\mathit{x}}}}
\newcommand\bb{\boldsymbol{\mathit{b}}}

\newcommand\bbhat{\widehat{\boldsymbol{\mathit{b}}}}

\newcommand\DD{\boldsymbol{\mathit{D}}}
\renewcommand\AA{\boldsymbol{\mathit{A}}}
\newcommand\BB{\boldsymbol{\mathit{B}}}
\newcommand\CC{\boldsymbol{\mathit{C}}}

\renewcommand\SS{\boldsymbol{\mathit{S}}}

\newcommand\II{\boldsymbol{\mathit{I}}}

\newcommand\LL{\boldsymbol{\mathit{L}}}

\newcommand\MM{\boldsymbol{\mathit{M}}}

\newcommand\yy{\boldsymbol{\mathit{y}}}
\newcommand\vv{\boldsymbol{\mathit{v}}}

\newcommand\vvhat{\widehat{\boldsymbol{\mathit{v}}}}

\newcommand{\defeq}{:=}

\renewcommand\P{\mathbb{P}}
\newcommand\E{\mathbb{E}}
\newcommand{\vertiii}[1]{{\left\vert\kern-0.25ex\left\vert\kern-0.25ex\left\vert #1 
    \right\vert\kern-0.25ex\right\vert\kern-0.25ex\right\vert}}

\newcommand{\approxbar}{\overset{\underline{\hspace{0.6em}}}{\approx}}
\newcommand{\Otil}{\widetilde{O}}

\newcommand{\R}{\mathbb{R}}
\newcommand{\Z}{\mathbb{Z}}

\newcommand{\eps}{\epsilon}

\newcommand{\0}{\boldsymbol{0}}
\newcommand{\1}{\boldsymbol{1}}

\newcommand{\calC}{\mathcal{C}}
\newcommand{\nnz}{\text{nnz}}

\begin{document}

\pagenumbering{roman}

\maketitle

\begin{abstract}
We present an algorithm that given any invertible symmetric diagonally dominant M-matrix (SDDM), i.e., a principal submatrix of a graph Laplacian, $\LL$ and a nonnegative vector $\bb$, computes an entrywise approximation to the solution of $\LL \xx = \bb$ in $\Otil(m n^{o(1)})$ time with high probability, where $m$ is the number of nonzero entries and $n$ is the dimension of the system.
\end{abstract}

\clearpage
\tableofcontents
\clearpage

\pagenumbering{arabic}

\newcounter{algcounter}

\section{Introduction}

Graph Laplacians have been a central subject of study in computer science over the past few decades due to their many connections to graphs and random walks. The seminal work of \cite{ST04:journal} presented the first near-linear time (in number of nonzeros of the matrix) algorithm for computing the solution of system $\LL \xx = \bb$, where $\LL$ is a graph Laplacian or more generally a principal submatrix of a graph Laplacian, i.e., symmetric diagonally dominant M-matrices (SDDM). Many follow-up works contributed to improving many aspects of such near-linear solvers such as simplifying the algorithm \cite{KOSZ13,LS13,KS16}, improving the polylogarithmic factors in the running time \cite{KMST10,KMP11,LS13,PS14,KMP14,CKMPPRX14,JS21,KS16}, and developing parallel solvers with polylogarithmic depth and near-linear work \cite{BGKMPT11,PS14,KLPSS16}.

All of these algorithms have focused on normwise error bounds and output a solution $\widehat{\xx}$ such that
\[
\norm{\widehat{\xx} - \LL^{\dagger} \bb}_{\LL} \leq \epsilon \cdot \norm{\LL^{\dagger} \bb}_{\LL}.
\]
As pointed out recently by \cite{GNY25}, since the entries of the solution can vary exponentially, such normwise error bounds cannot recover small entries unless the error parameter $\epsilon$ is exponentially small in $n$, the size of the matrix. Since these algorithms are all iterative preconditioning approaches, they require about $\log(1/\epsilon)$ iterations and the numbers used in the algorithm need $\log(\kappa/\epsilon)$ bits of precision to guarantee convergence, where $\kappa$ is the condition number of the matrix. Taking $\epsilon$ exponentially small means the number of bit operations, which we refer to as the running time throughout the paper, of the algorithm is about $O(mn^2)$. 
\cite{GNY25} improved on this by presenting an algorithm with a running time of $\Otil(m \sqrt{n} \log^2 (U \epsilon^{-1} \delta^{-1}))$, where $U$ is a bound on the magnitude of entries and $1-\delta$ is the probability of success, that for any invertible SDDM matrix $\LL$ produces $\widehat{\xx}$ such that
\[
e^{-\epsilon} (\LL^{-1} \bb)_i \leq \widehat{\xx}_i \leq e^{\epsilon} (\LL^{-1} \bb)_i, ~~ \text{for all} ~~ i \in [n],
\]
which we denote by $\widehat{\xx} \approxbar_{\epsilon} \LL^{-1} \bb$. This raises the question of whether solving a linear system with the stronger notion of entrywise approximation is possible in almost-linear time. We answer this question in the affirmative.

\begin{restatable}{theorem}{MainTheorem}
\label{thm:main}
   There exists a randomized algorithm \callalg{SDDMSolve} such that, for any $\delta \in (0, 1)$ and any $\eps > (nU)^{-2^{\sqrt{\log n}}}$ (i.e., not exponentially small), any invertible SDDM matrix $\LL \in \Z^{n \times n}$ with $m$ nonzero integer entries in $[-U, U]$ and a nonnegative vector $\bb \in \Z^n$ with integer entries in $[0, U]$, with probability at least $1 - \delta$, computes the entrywise approximate solution
   \begin{align*}
       \xxtil \approxbar_{\epsilon} \LL^{-1} \bb
   \end{align*}
   whose entries are represented by $O(\log (nU/\eps))$-bit floating points, using\footnote{$\Otil$ suppresses the $\log \log (U \eps^{-1} \delta^{-1})$ factors.} 
   \[
   \Otil(m 2^{O(\sqrt{\log n})} \log(U) \log^{2} (U \eps^{-1} \delta^{-1}))
   \]
   bit operations.
\end{restatable}

Our algorithm first constructs a low-diameter cover by solving a collection of linear systems with normwise error bounds on random right-hand side vectors. We then use this low-diameter cover to \emph{predict} the scale of the entries of the solution in an adaptive manner, thereby achieving an entrywise approximate solution.

In the next subsection, we introduce our notation and review the relevant results from the literature. Equipped with these preliminaries, we provide an overview of our algorithm and the main ideas behind the proof of \Cref{thm:main} in \Cref{sec:technical}. We discuss related work in \Cref{sec:related}. In \Cref{sec:low-diam-cover}, we define our low-diameter cover and present an algorithm for constructing it with high probability. Finally, \Cref{sec:almost-lin-alg} presents our almost-linear–time algorithm for solving SDDM linear systems with entrywise approximation guarantees, together with the proofs of correctness and bit-complexity bounds. We conclude in \Cref{sec:conclusion} with a discussion of open problems and future directions.

\subsection{Preliminaries}
\label{sec:prelim}

\textbf{Notation.}  
We denote $\{1, 2, \dots, n\}$ by $[n]$. For an undirected graph $G = (V, E, w)$ and a vertex $v \in V$, let $N(v)$ denote the set of neighbors of $v$. For a subset of vertices $S \subseteq V$, we denote by $G(S)$ the subgraph of $G$ induced by $S$. We use $\Otil$ notation to suppress polylogarithmic factors in $n$ and polyloglog factors in $U \epsilon^{-1} \delta^{-1}$. Formally, $\Otil(f) = O(f \cdot \log(n \cdot \log(U \epsilon^{-1} \delta^{-1})))$. We also note that $\poly\log n$ factors are at most $2^{O(\sqrt{\log n})}$.

We use bold lowercase letters to denote vectors and bold uppercase letters to denote matrices. The vector of all ones is denoted by $\1$, and the $i$-th standard basis vector by $\ee^{(i)}$. We do not explicitly specify their dimensions when they are clear from context. For a matrix $\MM$, we denote its $(i, j)$-entry by $\MM_{ij}$. If $S$ and $T$ are subsets of its row and column indices, respectively, then $\MM_{ST}$ denotes the submatrix formed by the entries $\{ (i, j) : i \in S,\, j \in T \}$.

For $\vv \in \R^n$, let $\diag(\vv) \in \R^{n \times n}$ denote the diagonal matrix with the entries of $\vv$ on its diagonal. For $\MM \in \R^{n \times n}$, let $\diag(\MM) \in \R^{n \times n}$ denote the diagonal matrix obtained from $\MM$ by zeroing out all off-diagonal entries. The Moore–Penrose pseudoinverse of $\MM$ is denoted by $\MM^\dagger$.

\begin{definition}[RDDL and SDDM Matrices]
A matrix $\MM \in \R^{n \times n}$ is called an L-matrix if $\MM_{ij} \le 0$ for all $i \neq j$ and $\MM_{ii} > 0$ for all $i \in [n]$.  
It is \emph{row diagonally dominant} (RDD) if for every $i \in [n]$,  
\[
|\MM_{ii}| \ge \sum_{j \in [n] \setminus \{ i \}} |\MM_{ij}|.
\]
We call $\MM$ \emph{RDDL} if it is both an L-matrix and row diagonally dominant.  
An \emph{SDDM matrix} is an invertible symmetric RDDL matrix.
\end{definition}

In this paper, we focus exclusively on SDDM matrices. An equivalent characterization is that such matrices are principal submatrices of undirected graph Laplacians. We refer to the corresponding graph as the \emph{associated graph} of the SDDM matrix, defined as follows.

\begin{definition}[Associated Graph of an RDDL Matrix]
For an RDDL matrix $\MM \in \R^{n \times n}$, we define the associated weighted graph $G = ([n+1], E, w)$ by adding a dummy vertex $n+1$. For every pair $i, j \in [n]$ with $i \neq j$, we set the edge weight $w(i, j) = -\MM_{ij}$.  
Additionally, we set the edge weight $w(i, n+1) = \sum_{j \in [n]} \MM_{ij}$.
\end{definition}

The associated graph of an SDDM matrix is defined analogously.  
The solution of a linear system with an SDDM matrix is closely connected to random walks on its associated graph, which we define next.

\paragraph{Random walks.}  
Given an undirected weighted graph $G = ([n], E, w)$ with edge weights $w \in \R_{\ge 0}^{E}$, a random walk on $G$ moves from a vertex to one of its neighbors independently of all previous steps.  
If the walk is currently at vertex $i$, it moves to $j \in N(i)$ in the next step with probability
\[
\AA_{ij} := \frac{w(i, j)}{\sum_{k \in N(i)} w(i, k)}.
\]
In other words, $\AA$ is the transition matrix of the Markov chain associated with $G$.  
For unweighted graphs, this simplifies to a uniform transition probability $1 / |N(i)|$ for each neighbor.

\begin{definition}[Escape Probability]
The \emph{escape probability} $\P(s, t, p)$ denotes the probability that a random walk starting at vertex $s$ visits vertex $t$ before vertex $p$.
\end{definition}

The following result characterizes the relationship between the entries of the inverse of an invertible SDDM matrix and the escape probabilities in its associated graph.

\begin{lemma}[Lemma 2.7 of \cite{GNY25}]
\label{lemma:escape-prob}
Let $\LL \in \R^{n \times n}$ be an invertible RDDL matrix. Then, for all $s, t \in [n]$,
\[
\P(s, t, n+1) = \frac{(\LL^{-1})_{st}}{\LL_{tt} \, (\LL^{-1})_{tt}}.
\]
Furthermore, $(\LL^{-1})_{tt} > 0$ for every $t \in [n]$.
\end{lemma}

We next recall a few linear algebraic notations and definitions that we use throughout the paper.

\noindent \textbf{Norms.}  
The $\ell_1$-norm and $\ell_{\infty}$-norm of a vector $\vv \in \R^n$ are defined as
$\norm{\vv}_1 := \sum_{i \in [n]} \abs{\vv_i}$ and
$\norm{\vv}_{\infty} := \max_{i \in [n]} \abs{\vv_i}$, respectively.
The induced $\ell_1$ and $\ell_{\infty}$ norms for a matrix $\MM \in \R^{n \times n}$ are defined as
\begin{align*}
    \norm{\MM}_1 &:= \max_{\vv \in \R^n : \norm{\vv}_1 = 1} \norm{\MM \vv}_1,\\
    \norm{\MM}_{\infty} &:= \max_{\vv \in \R^n : \norm{\vv}_{\infty} = 1} \norm{\MM \vv}_{\infty}.
\end{align*}
Both vector and induced norms satisfy the triangle inequality and consistency property; for example,
$\norm{\AA \BB}_{\infty} \le \norm{\AA}_{\infty} \cdot \norm{\BB}_{\infty}$.
We also define $\norm{\AA}_{\max} := \max_{i, j} \abs{\AA_{ij}}$. The \emph{energy norm} with respect to a positive definite matrix $\LL$ is defined as  
\[
\norm{\vv}_{\LL} := \sqrt{\vv^\top \LL \vv}.
\]

\vspace{1em}

\noindent \textbf{Schur complement.}  
Given a block matrix
\[
    \MM
    =
    \begin{bmatrix}
        \AA & \BB\\
        \CC & \DD
    \end{bmatrix},
\]
the \emph{Schur complement} of $\MM$ with respect to $\AA$ is
$\text{Sc}(\MM, \AA) = \DD - \CC \AA^{-1} \BB$, assuming $\AA$ is invertible.
If both $\SS := \text{Sc}(\MM, \AA)$ and $\AA$ are invertible, then $\MM$ is invertible and
\begin{equation} \label{eq:block-Sc-inversion}
    \MM^{-1}
    =
    \begin{bmatrix}
        \AA^{-1} + \AA^{-1} \BB \SS^{-1} \CC \AA^{-1} & - \AA^{-1} \BB \SS^{-1}\\
        - \SS^{-1} \CC \AA^{-1} & \SS^{-1}
    \end{bmatrix}.
\end{equation}

\paragraph{Solving SDDM linear systems with normwise error bounds.}  
The following result by Spielman and Teng~\cite{ST04:journal} was a breakthrough, providing a near-linear–time algorithm for solving SDDM linear systems with normwise error guarantees.

\begin{lemma}[Theorem 5.5 of \cite{ST04:journal}]
\label{lemma:near-linear-solver}
Let $\LL \in \R^{n \times n}$ be an SDDM matrix with $m$ nonzero entries, and let $\bb \in \R^{n}$ have integer entries in $[-U, U]$.  
For any $\epsilon > 0$ and $\delta \in (0,1)$, there exists an algorithm that, with probability at least $1-\delta$ and using  
$\Otil(m \log(U \epsilon^{-1} \delta^{-1}) \log(\epsilon^{-1}))$ bit operations, computes a vector $\widetilde{\xx}$ such that
\[
\norm{\widetilde{\xx} - \LL^{\dagger} \bb}_{\LL}
\leq
\epsilon \cdot \norm{\LL^{\dagger}\bb}_{\LL}.
\]
Moreover, the entries of $\widetilde{\xx}$ can be represented using $O(n \log(U \epsilon^{-1}))$ bits in fixed-point arithmetic.
\end{lemma}

The above result bounds the error with respect to the energy norm.  
For our purposes, it is often more convenient to work with the Euclidean norm.  
\Cref{cor:near-linear-solver-SDDM}, which follows directly from \Cref{lemma:near-linear-solver} together with the following bound on the condition number of invertible SDDM matrices, provides such an error guarantee in the Euclidean norm.

\begin{lemma}[Corollary 2.5 of \cite{GNY25}]
\label{lemma:norm2-inv}
Let $\LL$ be an $n \times n$ invertible SDDM matrix with integer entries in $[-U, U]$. Then
\[
\norm{\LL^{-1}}_2
=
\frac{1}{\lambda_{\min}(\LL)}
<
n^2.
\]
\end{lemma}

\begin{corollary}[Corollary 3.2 of \cite{GNY25}]
\label{cor:near-linear-solver-SDDM}
Let $\LL \in \R^{n \times n}$ be an invertible SDDM matrix with $m$ nonzero entries, and let $\bb \in \R^{n}$ have integer entries in $[-U, U]$.  
For any $\epsilon > 0$ and $\delta \in (0,1)$, there exists an algorithm that, with probability at least $1-\delta$ and using  
$\Otil(m \log^2(U \epsilon^{-1} \delta^{-1}))$ bit operations, computes a vector $\widetilde{\xx}$ such that
\[
\norm{\widetilde{\xx} - \LL^{\dagger} \bb}_{2}
\leq
\epsilon \norm{\bb}_{2}.
\]
Moreover, the entries of $\widetilde{\xx}$ can be represented using $O(\log(n U \epsilon^{-1}))$ bits in fixed-point arithmetic.
\end{corollary}

\subsection{Technical Overview}
\label{sec:technical}

Our approach builds on the \emph{threshold decay (TD)} framework introduced in \cite{GNY25}. We first recall the main components of this framework and then explain how we adapt it to obtain an almost-linear-time algorithm for solving SDDM linear systems with entrywise approximation guarantees.

\paragraph{Threshold decay (TD) framework.}
The TD framework relies on three key observations:
\begin{enumerate}
    \item A solution with normwise error guarantees (for a polynomially small error parameter) also provides entrywise approximations for the large coordinates of the solution. Moreover, every solution to a linear system has at least one large entry.
    
    \item Given a subset of entries $F \subseteq [n]$, if we substitute an entrywise approximation for the entries in $F$ into the system, the resulting smaller linear system is only mildly affected. More precisely, substituting an $\exp(\epsilon)$-entrywise approximation for $F$ changes the solution on $[n] \setminus F$ by at most a multiplicative factor of $\exp(\epsilon)$.
    
    \item Removing small entries of the solution introduces only a limited error. Specifically, if $C \subseteq [n]$ satisfies $(\LL^{-1}\bb)_i < \theta$ for all $i \in C$, then with $F = [n] \setminus C$,
    \[
    \Bigl\|\begin{bmatrix}
        (\LL_{FF})^{-1}\bb_F ; \boldsymbol{0}
    \end{bmatrix} - \LL^{-1}\bb\Bigr\|_2 \le n^3 U \theta.
    \]
\end{enumerate}

The first two observations imply that normwise guarantees can be leveraged to recover entrywise approximations for some coordinates of the solution, which can then be substituted back to form a smaller system. Repeating this process allows us to iteratively recover approximations for progressively smaller entries. The third observation shows that if we knew in advance which entries were small, we could safely remove them and focus computation on a reduced system that preserves accuracy for the large entries.

These insights naturally motivate the following iterative algorithm. We maintain three disjoint subsets of indices partitioning $[n]$:  
\emph{(i)} a \emph{solved set} $P$ containing entries whose entrywise approximations have been computed,  
\emph{(ii)} an \emph{inactive set} $Q$ of entries deemed small enough to omit, and  
\emph{(iii)} an \emph{active set} $R$ of remaining entries participating in the current linear system.  
At each iteration, the algorithm maintains a threshold parameter $\theta$ and proceeds as follows:
\begin{enumerate}
    \item Add to $R$ all entries in $Q$ whose magnitude is at least~$\theta$. Update $Q$ as $Q \gets [n] \setminus (R \cup P)$.
    \item Solve the linear system with coefficient matrix $\LL_{RR}$ and recover entrywise approximations for the entries in $R$ that are large relative to $\theta$. Move these entries from $R$ to $P$.
    \item Substitute the newly solved entries into the system to obtain an updated right-hand-side vector~$\bb$.
    \item Reduce the threshold $\theta$ by a polynomial factor in $nU$ and repeat.
\end{enumerate}
We remark that the first step is the most technically challenging and constitutes our main contribution.
To make the algorithm efficient, we must restrict the system solve to the active set 
$
R
$.
However, this restriction provides no direct information about the large entries in 
$
Q
$.
The core difficulty is to identify these large entries without explicitly solving on 
$
Q
$.
Our key idea is to achieve this through carefully designed \emph{predictions}.

\paragraph{TD with prediction.}
Achieving almost-linear running time requires keeping the active set small on average across iterations. This, in turn, depends on having good \emph{a priori} predictions for the magnitudes of the solution entries. A notable instance where such predictions are available is the computation of the inverse of an SDDM matrix (see \cite{GNY25}). Specifically, for any pair $(i,j)$ with $\LL_{ij} \ne 0$, one has
\[
\frac{1}{nU}\,\LL^{-1}\ee^{(j)} \;\le\; \LL^{-1}\ee^{(i)} \;\le\; nU\,\LL^{-1}\ee^{(j)},
\]
implying that the $i$-th column of $\LL^{-1}$ provides reliable predictions for the $j$-th column. Using this property, \cite{GNY25} obtained an $\tilde{O}(mn)$-time algorithm for computing the inverse of an SDDM matrix. Once one column of $\LL^{-1}$ is computed, each subsequent column can be obtained in $\tilde{O}(m)$ time.

However, this prediction only applies when inverting the entire matrix, i.e., solving $O(n)$ systems. For a single system, we must establish an \emph{upper bound} on the small entries in order to identify which entries can be safely ignored when extracting the large ones. Directly using the inequality gives only \emph{lower bounds} on the small entries, making it useless here.

\paragraph{Low-diameter (LD) cover.}
Obtaining accurate predictions for a single linear system
requires a more sophisticated approach. In this paper, we obtain such predictions by first computing a \emph{low-diameter cover} for the graph associated with~$\LL$, where we define the distance between two vertices $i$ and $j$ in terms of their corresponding entry in the inverse:
\[
D_{\LL}(i,j) := -\log_{nU}\!\big((\LL^{-1})_{ij}\big) + 2.
\]
An important property of this distance function is that it satisfies the triangle inequality.

We construct a collection of pairs $(V, W)$ of \emph{inner and outer balls}, where $V \subseteq W \subseteq [n]$, satisfying the following properties:
\begin{enumerate}[label=(\roman*)]
    \item Every vertex is contained in at least one inner ball.
    \item Each vertex is contained in at most $n^{o(1)}$ outer balls.
    \item For any pair $i,j$ within the same outer ball, $D_{\LL}(i,j) \le n^{o(1)}$.
    \item If $i$ lies in an inner ball and $j$ lies outside its corresponding outer ball, then $D_{\LL}(i,j) \ge n^{\omega(1)}$.
\end{enumerate}

We next explain how such a low-diameter cover can be leveraged to solve the linear system in almost-linear time, and subsequently describe how to compute such a cover efficiently.

\paragraph{Almost linear-time solver using LD cover.}
In iteration $t$ of the TD framework, let $P^{(t)}$ denote the solved set where the entrywise approximations of the entries have been computed. We now solve on the reduced system on $[n] \setminus P^{(t)}$ (obtained by plugging in the solved entries in $P^{(t)}$)
\begin{align*}
\LL_{[n] \setminus P^{(t)},\, [n] \setminus P^{(t)}} \xx = \widehat{\bb}^{(t)}
\end{align*}
where $\widehat{\bb}^{(t)} \approxbar     \bb_{[n] \setminus P^{(t)}} - \LL_{[n] \setminus P^{(t)},\, P^{(t)}} \widetilde{\xx}^{(t)}
$ is the right-hand-side vector of the reduced system.

It can be shown that (see Lemma~3.8 of~\cite{GNY25}) the nonzero entries of $\widehat{\bb}^{(t)}$ differ by at most a polynomial (in~$U$) factor. This allows us to invoke the near-linear-time solver of~\cite{ST04:journal} to compute a solution with normwise error guarantees in each iteration.

Let $I^{(t)}$ denote the set of nonzero entries of $\widehat{\bb}^{(t)}$.
The entries indexed by $I^{(t)}$ in the solution 
$
(\LL_{[n] \setminus P^{(t)},\, [n] \setminus P^{(t)}})^{-1}  \widehat{\bb}^{(t)}
$
are large---this follows directly from a simple application of \Cref{lemma:escape-prob}. We use $I^{(t)}$ as a seed to construct the active set, ensuring that the inactive set contains only small entries of the solution.

We define the \emph{boundary-expanded set} (see \Cref{def:boundary-expanded-set}) of $I^{(t)}$ as the union of inner balls whose corresponding outer balls intersect $I^{(t)}$. Due to the definition of the LD cover, if a vertex $u \in [n]$ does not belong to the boundary-expanded set of $I^{(t)}$, the corresponding outer ball of the inner ball that contains $u$ does not intersect $I^{(t)}$, then the probability distance between $u$ and any $v \in I^{(t)}$ is at least $n^{\omega(1)}$. This implies that the entry at $u$ is small and can be omitted in iteration~$t$ of the TD framework. Therefore, the active set in iteration~$t$, denoted by~$H^{(t)}$, is defined as the intersection of $[n] \setminus P^{(t)}$ and the boundary-expanded set of $I^{(t)}$.

Moreover, we show that the total size of all $H^{(t)}$'s satisfies
    \[
    \sum_t |H^{(t)}| \le n^{1+o(1)}.
    \]
By properties of the TD framework, every vertex is added to and deleted from the active set exactly once. When a vertex $u$ is added to the active set, the distance $d$ to the solved set is bounded by the radius of an outer ball in the LD cover, i.e., $n^{o(1)}$. This gives a lower bound on the $u$-th entry of the solution $(\LL^{-1} \bb)_u \ge \theta \cdot (nU)^{-d}$, where $\theta$ is the current threshold. Since the threshold is reduced by a polynomial factor in $nU$ in each iteration, the vertex will be removed from the active set and added to the solved set after $O(d) = n^{o(1)}$ iterations, concluding the bound on total size of $H^{(t)}$'s.

Maintaining the sets $I^{(t)}$ and $H^{(t)}$ requires additional care to ensure almost-linear running time:
\begin{enumerate}
    \item It can be shown that the right-hand-side vector $\bbhat^{(t)}$ undergoes only $O(m+n)$ updates in total. Hence, the sets $I^{(t)}$ can be maintained in $O(m+n)$ time overall.
    \item The set $H^{(t)}$ is maintained by keeping a counter for each pair in the LD cover that tracks the number of vertices in $I^{(t)}$ contained in the corresponding outer ball. Then $H^{(t)}$ is the union of the inner balls of all pairs with nonzero counters. Since the number of updates to $I^{(t)}$ is bounded and each vertex is covered by a limited number of outer balls, these counters can be updated in almost-linear time. Moreover, since
    \[
    \sum_t |H^{(t)}| \le n^{1+o(1)},
    \]
    and each vertex is covered by at most $n^{o(1)}$ outer balls,
    using the property of our TD framework, the total time to construct all sets $H^{(t)}$ is at most $n^{1+o(1)}$ over the course of the algorithm.
\end{enumerate}

\paragraph{Constructing the LD cover.}
To explain the high-level ideas of our algorithm for constructing the LD cover, we assume exact solutions to linear systems. However, in \Cref{sec:low-diam-cover}, we carefully address bit complexity and the errors introduced when solving linear systems approximately.

Consider the solution to the linear system $\LL \xx = \ee^{(i)}$. Let
\[
V = \{j \in [n] : (\LL^{-1}\ee^{(i)})_j \ge 2^{-\sqrt{\log n}}\}
\quad\text{and}\quad
W = \{j \in [n] : (\LL^{-1}\ee^{(i)})_j \ge 4^{-\sqrt{\log n}}\}.
\]
Then $V \subseteq W$. By the triangle inequality for $D_{\LL}$, one can verify that for the inner ball $V$ and the outer ball $W$, properties (iii) and (iv) of the LD cover are satisfied. Moreover, $i$ is contained in $V$, and therefore property (i) is satisfied for $i$. Thus, one may construct a pair of balls for each $i \in [n]$ by considering the large entries of the solution to $\LL \xx = \ee^{(i)}$.

This naïve approach encounters two main issues: (1) it requires solving $n$ different linear systems; and (2) each vertex may belong to too many outer balls. To address these issues, we instead attempt to construct outer/inner ball pairs for multiple vertices simultaneously by solving a linear system with right-hand side $\boldsymbol{1}_S = \sum_{i\in S} \ee^{(i)}$. By linearity, the large entries of the solution $\LL \xx = \boldsymbol{1}_S$ roughly give the \emph{union} of balls for $i \in S$.
An important observation is that if $i$ is far (in $D_{\LL}$ distance) from all $j \in S \setminus \{i\}$, then the pair of balls corresponding to $i$ form an isolated connected component from the union of the rest, and therefore can be recovered from the solution.

To obtain such an $S$, we select a random subset of $[n]$, where each vertex is included independently. Moreover, instead of using the fixed thresholds $2^{-\sqrt{\log n}}$ and $4^{-\sqrt{\log n}}$ for the inner and outer balls, we consider multiple thresholds that are exponential in a set of distance parameters $d_i$ to ensure that vertices in $S$ remain sufficiently far from one another.

Therefore, we construct the LD cover by solving a \emph{collection of linear systems} with coefficient matrix~$\LL$ and random right-hand-side vectors, and by recovering inner/outer balls from the large entries of the resulting solutions.
We consider $O(\sqrt{\log n})$ different probabilities~$p_j$ (for selecting vertices to set $S$) and $O(\sqrt{\log n})$ different distance thresholds~$d_i$ in a nested loop. For each pair~$(d_i, p_j)$, we run $n^{o(1)}$ independent iterations. In each iteration, every vertex~$u$ is included in~$S$ independently with probability~$p_j$, forming a random subset~$S$. We then compute a normwise approximate solution with error parameter exponential in $d_i$ to the linear system
\[
\LL \xx = \boldsymbol{1}_S.
\]

Let $T$ denote the set of vertices corresponding to the \emph{large} entries (i.e., entries larger than $\exp(-d_i/2)$) in the resulting solution. If a connected component of the subgraph induced by $T$ contains exactly one vertex of $S$, then we construct a pair of inner and outer balls $(V, W)$ on that connected component such that the distance between any vertex inside the inner ball and any vertex outside the outer ball is at least $n^{\omega(1)}$, and the distance between any two vertices within the outer ball is at most $n^{o(1)}$. We construct inner/outer ball pairs only from such connected components and add them to our cover.

This constraint on the connected components ensures that the vertex in $S$ contained in that component is far from all other vertices in $S$. Consequently, the construction satisfies properties~(iii) and~(iv) of the LD cover.

We show that for every vertex~$u \in [n]$, there exists a pair~$(d_i, p_j)$ such that, with probability at least~$n^{-o(1)}$, the random set~$S$ contains a vertex near~$u$ (that is, within its corresponding inner ball) and that vertex is the unique element of~$S$ in its connected component of the subgraph induced by~$T$, which is our requirement for constructing inner/outer balls as explained above. The repeated iterations in the inner loop amplify the probability of this event. Consequently, with high probability, every vertex is contained in at least one inner ball. Moreover, since in each iteration we only include disjoint outer balls, each vertex is covered by at most $n^{o(1)}$ outer balls. Hence, properties~(i) and~(ii) also hold.

\subsection{Related Work}
\label{sec:related}

Solving linear systems with entrywise approximation guarantees dates back to the 1980s and 1990s.  
From a theoretical point of view, Higham~\cite{H90} observed that fast matrix multiplication methods such as Strassen’s~\cite{S69} cannot achieve entrywise approximations for matrix multiplication (which, via standard reductions, can be used for solving linear systems) due to their inherent reliance on subtractions, even when applied to nonnegative matrices.

On the other hand, empirical studies demonstrated that Gaussian elimination is significantly more stable when computing the stationary distribution of Markov chains~\cite{H87,GTH85,HP84} compared to its application to general matrices.  
These studies attributed the improved stability to the fact that Gaussian elimination, when applied to SDDM matrices, does not require subtractions and is therefore highly compatible with floating-point arithmetic.

Subsequent works further analyzed the error bounds and the number of floating-point operations required for computing the stationary distribution~\cite{O93} and for performing LU factorization~\cite{O96} with entrywise approximation guarantees.  
All of these algorithms require cubic time.  
Very recently,~\cite{GNY25} studied this problem for both SDDM and RDDL matrices with the goal of designing subcubic-time algorithms.  
In particular, for solving a linear system with an SDDM matrix, they presented an algorithm with running time $\Otil(m \sqrt{n} \log^2 (U \epsilon^{-1} \delta^{-1}))$.  
Moreover, for inverting an SDDM matrix, they provided an algorithm with running time $\Otil(m n \log^2 (U \epsilon^{-1} \delta^{-1}))$.

The recent work of Gao, Kyng, and Spielman~\cite{GKS23:arXiv} has demonstrated the superior practical performance of near-linear–time Laplacian solvers.  
We hope that our approach can further advance practical solvers, particularly in settings where edge weights vary widely---for example, in solving linear systems that arise as subroutines in interior-point methods for flow problems.

There also exist algorithms in the literature for decomposing graphs into components with small diameter or covering them with low-diameter components, with applications to shortest-path problems and parallel SDDM solvers~\cite{C98,BGKMPT11}.  
However, these low-diameter decompositions and covers differ substantially from ours.  
To begin with, our covers are defined using pairs of inner and outer balls, with the guarantee that vertices inside an inner ball are far from those outside the corresponding outer ball.  
In contrast, the decompositions and covers in the literature are defined using single clusters or balls of small diameter.  
Moreover, our notion of distance is based on the \emph{probability distance}
\[
D_{\LL}(i,j) := -\log_{nU}\!\big((\LL^{-1})_{ij}\big) + 2,
\]
as opposed to the standard (weighted or unweighted) graph distance. Most importantly, our access to the distance $D_{\LL}(\cdot,\cdot)$ is highly restricted since it requires computing entries of the inverse. We address this limitation via a carefully designed oracle query, obtained by solving the system on a random vector with normwise error guarantees.

\section{Low-Diameter Cover}
\label{sec:low-diam-cover}

In this section, we introduce our definition of a low-diameter cover and present an algorithm for constructing it with high probability using SDDM solvers with normwise error bounds as a subroutine.  
We note that the existence of such low-diameter covers is nontrivial, and the proof of correctness of our algorithm also serves as a constructive proof of their existence.

The diameters of the components in our cover are measured with respect to a distance function that we call the \emph{probability distance}.  
We define this notion and establish some of its useful properties in \Cref{sec:prob-dist}.  
In \Cref{sec:low-dim-cover-alg}, we formally define our low-diameter cover and present an algorithm for constructing it.

\subsection{Probability Distance}
\label{sec:prob-dist}

The following distance function, defined in terms of the entries of the inverse of an SDDM matrix, forms the backbone of our low-diameter cover construction.

\begin{definition}[Probability distance] 
\label{def:probability-distance}
Let $\LL \in \Z^{n \times n}$ be an invertible SDDM matrix with integer entries in $[-U, U]$.  
We define the \emph{probability distance} with respect to $\LL$ as
\[
    D_{\LL}(i,j) := -\log_{nU}\!\big((\LL^{-1})_{ij}\big) + 2,
    \quad \text{for all } i,j \in [n].
\]
\end{definition}

The following result establishes some basic properties of the probability distance, including the triangle inequality and the fact that the distance between two adjacent vertices is small.

\begin{claim}[Properties of the probability distance]
\label{claim:dist-prop}
Let $\LL \in \Z^{n \times n}$ be an invertible SDDM matrix with integer entries in $[-U, U]$.  
Then, for all $i, j, k \in [n]$, the probability distance $D_{\LL}$ satisfies:
\begin{enumerate}
    \item $D_{\LL}(i,i) \ge 0$.
    \item $D_{\LL}(i,j) = D_{\LL}(j,i)$.
    \item $D_{\LL}(i,k) \le D_{\LL}(i,j) + D_{\LL}(j,k)$.
    \item For adjacent vertices $i,j$ (i.e., $\LL_{ij} \neq 0$), we have $D_{\LL}(i,j) \le 4$.
\end{enumerate}
\end{claim}

\begin{proof}
    The first property follows from \Cref{lemma:norm2-inv} which states that the largest eigenvalue of $\LL^{-1}$ is at most $n^2$.
    The second property holds because $\LL$ is symmetric.
    For the fourth property, let $\DD \defeq \diag(\LL)$ and $\AA = \DD^{-1}(\DD - \LL)$ and $i,j$ be adjacent.
    By \Cref{lemma:escape-prob}
    \begin{align*}
        (\LL^{-1})_{ij}
        &=
        \P(i,j,n+1) \LL_{jj} (\LL^{-1})_{jj}\\
        &=
        \P(i,j,n+1) \LL_{jj} ((\DD (\II - \AA))^{-1})_{jj}\\
        &= \P(i,j,n+1) \LL_{jj} \DD^{-1}_{jj} ((\II - \AA)^{-1})_{jj}\\
        & = 
        \P(i,j,n+1) \LL_{jj} \DD^{-1}_{jj} (\II + \AA + \AA^{2} + \cdots)_{jj}
        \\
        &\ge U^{-2},
    \end{align*}
    where the inequality holds because $\P(i,j,n+1), \DD^{-1}_{jj} \geq U^{-1}$ and $\LL_{jj}, (\II + \AA + \AA^{2} + \cdots)_{jj} \geq 1$. Therefore
    \[
    D_{\LL}(i,j) = -\log_{nU} ( (\LL^{-1})_{ij}) + 2 \leq -\log_{nU} ( U^{-2}) + 2 \leq 4.
    \]
    Now we prove the third property, i.e., the triangle inequality. By \Cref{lemma:norm2-inv},
    \[
    (\LL^{-1})_{jj} \leq n^2.
    \]
    Therefore since $\LL_{jj} \leq U$,
    \begin{align}
    \label{eq:diag-mult-less-2}
    \log_{nU}(\LL_{jj}(\LL^{-1})_{jj}) \leq 2.
    \end{align}
    Now note that due to Markov property,
    \[
    \P(i,j,n+1) \cdot \P(j,k,n+1) \leq \P(i,k,n+1).
    \]
    Therefore
    \[
    (\P(i,j,n+1) \cdot \LL_{jj} (\LL^{-1})_{jj}) \cdot (\P(j,k,n+1) \cdot \LL_{kk} (\LL^{-1})_{kk}) \leq (\P(i,k,n+1)\cdot \LL_{kk} (\LL^{-1})_{kk})) \cdot (\LL_{jj} (\LL^{-1})_{jj}).
    \]
    Thus by \Cref{lemma:escape-prob},
    \[
    (\LL^{-1})_{i,j} \cdot (\LL^{-1})_{j,k} \le (\LL^{-1})_{i,k} \cdot (\LL_{jj} (\LL^{-1})_{jj}).
    \]
    Taking the logarithm and using \Cref{eq:diag-mult-less-2}, we have
    \[
    \log(\LL^{-1})_{i,j} + \log (\LL^{-1})_{j,k} \leq \log (\LL^{-1})_{i,k} + \log (\LL_{jj} (\LL^{-1})_{jj}) \leq \log (\LL^{-1})_{i,k} + 2. 
    \]
    Therefore
    \[
    D_{\LL}(i,j) + D_{\LL}(j,k) = -\log(\LL^{-1})_{i,j} - \log (\LL^{-1})_{j,k} + 4 \geq - \log (\LL^{-1})_{i,k} - 2 + 4 = D_{\LL}(i,k).
    \]
\end{proof}

Finally, we conclude this part by showing that the principal submatrices of an invertible SDDM matrix $\LL$ induce larger distances than $\LL$ itself.  
We will use this property in our almost-linear–time algorithm to establish that if two vertices are far from each other with respect to $\LL$, then they remain far with respect to any principal submatrix of $\LL$.

\begin{lemma}[Monotonicity of distance]
\label{lem:distance-monotonicity}
Let $\LL \in \Z^{n \times n}$ be an invertible SDDM matrix, and let $\LL_{S,S}$ denote the principal submatrix indexed by $S \subseteq [n]$.  
Then, for all $i, j \in S$, we have
\[
    D_{\LL}(i,j) \le D_{\LL_{S,S}}(i,j).
\]
\end{lemma}

\begin{proof}
    Let $\DD \defeq \diag(\LL)$ and
    \begin{align*}
        \AA \defeq \DD^{-1}(\DD - \LL) = \II - \DD^{-1} \LL,
    \end{align*}
    which represents the transition matrix of the associated graph to $\LL$. Then $\LL = \DD(\II - \AA)$. Therefore, $\LL^{-1} = (\II - \AA)^{-1} \DD^{-1}$. We also have $\LL_{S,S} = \DD_{S,S} (\II - \AA)_{S,S}$ and $(\LL_{S,S})^{-1} = ((\II - \AA)_{S,S})^{-1} (\DD_{S,S})^{-1}$.
    By invertibility of $\LL$,
    \begin{align*}
        \left( \II - \AA \right)^{-1}
        &=
        \II + \AA + \AA^{2} + \cdots,\\
        \left( \II_{S,S} - \AA_{S,S} \right)^{-1}
        &=
        \II + \AA_{S,S} + (\AA_{S,S})^2 + \cdots.
    \end{align*}
    Consequently, for any $i,j \in S$ we have
    \begin{align*}
        \left( \LL^{-1} \right)_{ij}
        &=
        \frac{1}{\LL_{jj}} \sum_{k=0}^{\infty} \left( \AA^k \right)_{ij},
        \quad
        \left( (\LL_{S,S})^{-1} \right)_{ij}
        =
        \frac{1}{\LL_{jj}} \sum_{k=0}^{\infty} \left( (\AA_{S,S})^k \right)_{ij}.
    \end{align*}
    It follows from the nonnegativity of $\AA$ that
    $\left( \AA^k \right)_{ij} \ge \left( (\AA_{S,S})^k \right)_{ij}$ for all $k$, and thus
    $(\LL^{-1})_{ij} \ge ((\LL_{S, S})^{-1})_{ij}$.
    The result then follows from the definition of the distance.
\end{proof}

\subsection{Constructing a Cover with Low Probability Distance Diameter}
\label{sec:low-dim-cover-alg}

In this section, we first present the definition of our low-diameter cover, which is characterized by three parameters: $r_{\text{in}}$, $r_{\text{out}}$, and $\alpha$.  
The parameter $\alpha$ serves as an upper bound on the number of outer balls that cover a single vertex.  
$r_{\text{in}}$ is an upper bound on the diameter of the outer balls, and consequently also on the diameter of the inner balls.  
Finally, $r_{\text{out}}$ provides a lower bound on the distance between vertices inside an inner ball and those outside its corresponding outer ball.

\begin{definition}[$(r_{\text{in}}, r_{\text{out}}, \alpha)$-cover]
\label{def:low-diam-cover}
    Given an $n$-by-$n$ invertible SDDM matrix $\LL$ with integer entries in $[-U,U]$, we say a set of pairs $\mathcal{C} = \{(V_1, W_1), (V_2, W_2), \dots, (V_k, W_k)\}$, with $V_i, W_i \subseteq [n]$, for all $i \in [k]$, is a $(r_{\text{in}}, r_{\text{out}}, \alpha)$-cover for $\LL$ if it satisfies the following where $D_{\LL}$ is a probability distance with respect to $\LL$.
    \begin{enumerate}
        \item $V_i \subseteq W_i \subseteq [n]$ for all $i \in [k]$.
        \item For every $u \in [n]$, there exists $i\in [k]$ such that $u \in V_i$.
        \item For every $u \in [n]$, $\abs{\{i\in[k]: u \in W_i\}} \leq  \alpha$.
        \item For any $i \in [k]$, and any $u, v \in W_i$, $D_{\LL}(u,v) \le r_{\text{in}}$. 
        \item For any $i \in [k]$, and any $u \in V_i$, $v \in [n] \setminus W_i$, $D_{\LL}(u, v) > r_{\text{out}}$.
    \end{enumerate}
    We colloquially refer to $V_i$'s as inner balls and $W_i$'s as outer balls. 
    
\end{definition}

The following result states that a low-diameter cover with suitable values of the parameters $r_{\text{in}}$, $r_{\text{out}}$, and $\alpha$ exists and can be computed in almost-linear time with high probability.  
The remainder of this section is devoted to proving this theorem.

\begin{figure}[!t]
    \begin{algbox}
        \textbf{\uline{\textsc{LowDiamConstruct}}}: Constructing a low-diameter cover.
        \\
        \uline{Input}: $\LL$, an $n \times n$ invertible SDDM matrix with integer entries in $[-U, U]$, and $\delta \in (0, 1)$ \\
        \uline{Output:} $\{(V_i, W_i)\}_i$ with properties stated in \Cref{def:low-diam-cover} \\
    \begin{enumerate}
    \item Let $\ell = \lceil \sqrt{\log n} \rceil + 3$, for $i \in [\ell]$, $d_{i} = \frac{4^{\ell}}{2^{i-1}}$, $p_{i} = \min\{\frac{1}{n} \cdot 2^{\ell (i-1)}, 1\}$.
    \item Let $M$ be the smallest power of two larger than $nU$.
    \item Set $\mathcal{C} = \emptyset$.
    \item For $i \in [\ell]$, $j \in [\ell]$:

    \begin{enumerate}
        \item \label{alg-step:low-diam-inner-loop} Repeat for $B := 6 \cdot 16^{\ell} \cdot \lceil \log( n/ \delta) \rceil$ times: 
        \begin{enumerate}
            \item Sample a set $S$ by including each element of $[n]$ independently with probability $p_{j}$.
            \item \label{alg-step:low-diam-l-solve} Compute a vector $\yy$ such that with probability $1-\delta/(2 \ell^2 B)$ satisfies that $\norm{\yy - \LL^{-1} \boldsymbol{1}_S}_{2} \leq M^{-2d_i}$. E.g., using \Cref{cor:near-linear-solver-SDDM}.
            \item Let $T =\{k \in [n]: \yy_{k} \geq M^{-d_i/2}\}$ and $G_T$ be the induced subgraph on $T$ of the graph corresponding to $\LL$.
            \item \label{alg-step:low-diam-construct-balls} For each $v \in S$, let $C_v$ be the connected component of $G_T$ that contains $v$. If for all $w \in S \setminus \{v\}$, $C_v\cap C_w = \emptyset$, then we set 
            \[
            V = \{k \in C_v: \yy_k \geq M^{-d_i/4} - M^{-2d_i}\}
            \]
            and 
            \[
            W= \{k \in C_v: \yy_k \geq M^{-(d_i/2)+2}\},
            \]
            and add $(V,W)$ to $\mathcal{C}$.
        \end{enumerate}
    \end{enumerate}
    \item Return $\mathcal{C}$
    \end{enumerate}
    
    \end{algbox}
    \caption{low-diameter construction}
    \labelalg{LowDiamConstruct}
\end{figure}

\begin{theorem}
\label{thm:construct-low-diam-cover}
Let $U > 32$ and $n>10$.
There exists an algorithm that given $\delta \in (0,1)$ and an $n$-by-$n$ invertible SDDM matrix $\LL$ with integer entries in $[-U,U]$, with probability at least $1-\delta$, computes an $(r_{\text{in}}, r_{\text{out}}, \alpha)$-cover $\mathcal{C}$ for $\LL$ with 
\[
r_{\text{in}}=2^{2\ell+1}, ~~ r_{\text{out}} = 2^{ \ell-2}, ~~ \alpha = 6 \ell^2 \cdot 16^{\ell} \cdot \lceil \log(n/\delta)\rceil, ~~ \ell =\lceil \sqrt{\log n} \rceil + 3
\]
in $m 2^{O(\sqrt{\log n})} \log(U) \log(\delta^{-1})\log(U \delta^{-1})$ bit operations.
\end{theorem}
\begin{proof}
We first prove that the output of \callalg{LowDiamConstruct} satisfies all the properties in \Cref{def:low-diam-cover} with probability at least \(1 - \delta\). With probability at least $1-\delta/2$, all the linear systems in Step \ref{alg-step:low-diam-l-solve} of \callalg{LowDiamConstruct} are solved correctly. Therefore for the rest of the proof, we assume that we have the error guarantee of Step \ref{alg-step:low-diam-l-solve}.

\begin{enumerate}
    \item Since $U>32$ and $n>10$, we have $d_i \geq 32$ and $M>320$. Therefore, $M^{-d_i/4} - M^{- 2 d_i} < M^{-(d_i/2)+2}$, and
\[
\{k \in C_u : \yy_k \ge M^{-d_i/4} - M^{- 2 d_i}\} \subseteq \{k \in C_u : \yy_k \ge M^{-(d_i/2)+2}\},
\]
for all $(V, W) \in \mathcal{C}$ we have $V \subseteq W$.

    \item Let $u \in [n]$.  First we show that there exists $i \in [\ell]$ such that 
    \begin{align}
    \label{eq:good-i}
    \abs{\{k \in [n]: (\LL^{-1} \ee^{(u)})_{k} \geq M^{-2d_i}\}} \leq 8^{\ell} \abs{\{k \in [n]: (\LL^{-1} \ee^{(u)})_{k} \geq M^{-d_i/4}\}}.
    \end{align}
    Suppose this does not hold for any $i \in [\ell]$. Then since $8 d_{i+3}=2d_{i+1} = d_{i}$, for $i\in [\ell-3]$, we have
    \[
    \abs{\{k \in [n]: (\LL^{-1} \ee^{(u)})_{k} \geq M^{-d_{i}}\}} > 8^{\ell} \abs{\{k \in [n]: (\LL^{-1} \ee^{(u)})_{k} \geq M^{-d_{i+3}}\}}.
    \]
    Therefore
    \[
    \abs{\{k \in [n]: (\LL^{-1} \ee^{(u)})_{k} \geq M^{-d_{1}}\}} > 8^{\ell \cdot (\lceil \ell/3 \rceil - 1)} \abs{\{k \in [n]: (\LL^{-1} \ee^{(u)})_{k} \geq M^{-d_{\ell}}\}}.
    \]
    Now note that 
    \[
    8^{\ell \cdot (\lceil \ell/3 \rceil - 1)} = 2^{3 \ell \cdot (\lceil \ell/3 \rceil - 1)} \geq 2^{\ell^{2} - 3\ell} \geq 2^{\log n} = n,
    \]
    where the last inequality follows from 
    \[
    \ell^{2} - 3\ell = \ell \cdot \lceil \sqrt{\log n} \rceil > \log n.
    \]
    Thus
    \begin{align}
    \label{eq:n-ratio-between-balls}
        \abs{\{k \in [n]: (\LL^{-1} \ee^{(u)})_{k} \geq M^{-d_{1}}\}} > n \cdot  \abs{\{k \in [n]: (\LL^{-1} \ee^{(u)})_{k} \geq M^{-d_{\ell}}\}}.
    \end{align}
    Moreover since $(\LL^{-1} \ee^{(u)})_u \geq 1/U$, $\abs{\{k \in [n]: (\LL^{-1} \ee^{(u)})_{k} \geq M^{-d_{\ell}}\}} \geq 1$. Additionally 
    \[
    \abs{\{k \in [n]: (\LL^{-1} \ee^{(u)})_{k} \geq M^{-d_{1}}\}} \leq n.
    \]
    However this contradicts \eqref{eq:n-ratio-between-balls}. Therefore there exists $i \in [\ell]$ such that 
    \[
    \abs{\{k \in [n]: (\LL^{-1} \ee^{(u)})_{k} \geq M^{-2d_i}\}} \leq 8^{\ell} \abs{\{k \in [n]: (\LL^{-1} \ee^{(u)})_{k} \geq M^{-d_i/4}\}}.
    \]
    Now since $p_1 = \frac{1}{n}$ and $p_n = 1$, there exists $j \in [\ell-1]$, such that $p_j < 1$
    \begin{align}
        \label{eq:good-j}
    \frac{1}{p_{j} \cdot 2^{\ell}} \leq \frac{1}{p_{j+1}} \leq |\{k \in [n]: (\LL^{-1} \ee^{(u)})_{k} \geq M^{-2 d_i}\}| \leq \frac{1}{p_{j}}.
    \end{align}
    Therefore for the $i$ and $j$ that satisfy \eqref{eq:good-i} and \eqref{eq:good-j}, respectively, for any iteration of the inner loop (Step \ref{alg-step:low-diam-inner-loop} of \callalg{LowDiamConstruct}), we have
    \[
    \frac{1}{2^{\ell}}\leq \E\left[\abs{S \cap \{k \in [n]: (\LL^{-1} \ee^{(u)})_{k} \geq M^{-2d_i}\}} \right] \leq 1.
    \]
    Let $\beta = \abs{\{k \in [n]: (\LL^{-1} \ee^{(u)})_{k} \geq M^{-2d_i}\}}$.
    We have
    \begin{align*}
    \P\left[\abs{S \cap \{k \in [n]: (\LL^{-1} \ee^{(u)})_{k} \geq M^{-2d_i}\}} = 1 \right] & = \beta p_j (1-p_j)^{\beta-1} 
    \\ & \geq \beta p_j e^{-p_{j} (\beta-1)/(1-p_j)} 
    \\ & \geq \beta p_j e^{-p_{j} \beta/(1-p_j)}
    \\ & \geq 2^{-\ell} e^{-1}.
    \end{align*}
    Therefore if $S_1,\ldots,S_{2^\ell}$ is the sampled set in $2^\ell$ different iterations of the inner loop (Step \ref{alg-step:low-diam-inner-loop} of \callalg{LowDiamConstruct}), with probability at least $1-e^{-e^{-1}} \geq 0.3$, there is a $q \in [2^{\ell}]$ such that
    \begin{align*}
    \abs{S_q \cap \{k \in [n]: (\LL^{-1} \ee^{(u)})_{k} \geq M^{-2d_i}\}} = 1.
    \end{align*}
    
    Now condition on $\abs{S_q \cap \{k \in [n]: (\LL^{-1} \ee^{(u)})_{k} \geq M^{-2d_i}\}} = 1$, the probability that 
    \[
    \abs{S_q \cap \{k \in [n]: (\LL^{-1} \ee^{(u)})_{k} \geq M^{-d_i/4}\}} = 1
    \]
    is
    \[
    \frac{\abs{\{k \in [n]: (\LL^{-1} \ee^{(u)})_{k} \geq M^{-d_i/4}\}}}{\abs{\{k \in [n]: (\LL^{-1} \ee^{(u)})_{k} \geq M^{-2d_i}\}}} \geq \frac{1}{8^{\ell}},
    \]
    where the inequality follows from \eqref{eq:good-i}. Therefore with probability at least $(1-e^{-1})\cdot 0.3 \geq \frac{1}{6}$ over $16^{\ell}$ iterations of the inner loop, we pick at least one $S$ such that 
    \begin{align}
        \label{eq:big-ball-empty}
    \abs{S \cap\{k \in [n]: (\LL^{-1} \ee^{(u)})_{k} \geq M^{-d_i/4}\}} = \abs{S \cap \{k \in [n]: (\LL^{-1} \ee^{(u)})_{k} \geq M^{-2d_i}\}} = 1.
    \end{align}
    Therefore over $B = 6 \cdot 16^{\ell} \cdot \lceil \log(2 n/ \delta) \rceil$ iterations, with probability at least
    $1-\frac{\delta}{2n}$, we pick at least one such $S$. 
    
    Let $v$ be the only element of $S \cap\{k \in [n]: (\LL^{-1} \ee^{(u)})_{k} \geq M^{-d_i/4}\}$. Therefore
    \begin{align}
    \label{eq:lower-bound-for-l-inv-uv}
    M^{-d_i/4} \leq (\LL^{-1} \ee^{(u)})_{v} \leq \sum_{k \in S} (\LL^{-1} \ee^{(u)})_{k} = (\LL^{-1} \boldsymbol{1}_S)_{u} \leq \yy_u + M^{-2d_i},
    \end{align}
    where the last inequality follow from $\norm{\yy - \LL^{-1} \boldsymbol{1}_S}_{2} \leq M^{-2d_i}$. Thus $\yy_u \geq M^{-d_i/4} -M^{-2d_i}$ and $u \in V$. We now show that for all $w \in S \setminus \{v\}$, $C_v\cap C_w = \emptyset$. 
    Since $U > 64$,
    \begin{align}
    \label{eq:log-M-nU-bounds}
    1 \leq \log_{nU} (M) \leq 1.2, ~\text{and} ~~ 0.8 \leq \log_{M} (nU) \leq 1.
    \end{align}
    
    Note that by \eqref{eq:big-ball-empty}, since $\{v\} = S \cap\{k \in [n]: (\LL^{-1} \ee^{(u)})_{k} \geq M^{-d_i/4}\}$, for any $w' \in S \setminus \{v\}$, we have
    $
    (\LL^{-1} \ee^{(u)})_{w'} < M^{-2d_i}
    $, and therefore
    \[
    -\log_M((\LL^{-1})_{uw'}) > 2d_i,
    \]
    and 
    \[
    D_{\LL}(u, w') = -\log_{nU}((\LL^{-1})_{uw'}) + 2 = \frac{-\log_M((\LL^{-1})_{uw'})}{\log_M (nU)} + 2 \geq 2 d_i + 2,
    \]
    where the inequality follows from \eqref{eq:log-M-nU-bounds}.
    By \eqref{eq:lower-bound-for-l-inv-uv},  $M^{-d_i/4} \leq (\LL^{-1})_{uv}$, and thus 
    \[
    \frac{d_i}{4} \geq - \log_{M} ((\LL^{-1})_{uv}).
    \]
    Therefore by \eqref{eq:log-M-nU-bounds},
    \begin{align}
    \nonumber
    0.3 d_i + 2 & = 1.2 \cdot \frac{d_i}{4} + 2 \geq - \log_{nU} (M) \log_{M} ((\LL^{-1})_{uv}) + 2 \\ & = - \log_{nU} ((\LL^{-1})_{uv}) + 2 = D_{\LL} (u,v).
    \end{align}
    For the sake of contradiction, suppose there exists $w \in S \setminus \{v\}$ with $w \in C_v$. Moreover let $v=k_0, k_1,\ldots, k_{p-1}, k_p=w$ be an arbitrary path between $v$ and $w$ that goes through $C_v$, i.e., $k_0,\ldots,k_p \in C_v$. By triangle inequality (\Cref{claim:dist-prop}), for all $g \in [p]$, we have
    \[
    D_{\LL}(u,k_{g}) \leq D_{\LL}(u,k_{g-1}) + D_{\LL}(k_{g-1},k_{g}).
    \]
    Moreover by \Cref{claim:dist-prop}, since for all $g \in [p]$, $k_{g-1}$ and $k_g$ are adjacent, $D_{\LL}(k_{g-1},k_g) \leq 4$.
    Therefore since $D_{\LL} (u,v) \leq 0.3 d_i + 2$, $D_{\LL} (u,w) \geq 2d_i + 2$, and $d_i \geq 8$, there exists $q \in [p]$ such that
    \[
    1.2\cdot d_i + 2 \leq D_{\LL}(u,k_q) \leq 1.2\cdot d_i + 6
    \]
    Therefore for all $w' \in S \setminus \{v\}$
    \[
    2d_i + 2\leq D_{\LL}(u,w') \leq D_{\LL}(u, k_q) + D_{\LL}(k_q, w') \leq 1.2 \cdot d_i + 6 + D_{\LL}(k_q, w'),
    \]
    which gives
    \[
    D_{\LL}(k_q, w') \geq 0.8 d_i - 4.
    \]
    Moreover by triangle inequality
    \[
    1.2\cdot d_i + 2 \leq D_{\LL}(u,k_q) \leq D_{\LL}(u, v) + D_{\LL}(v,k_q) \leq 0.3 d_i + 2 + D_{\LL}(v,k_q),
    \]
    which gives
    \[
    D_{\LL}(v,k_q) \geq 0.9 d_i.
    \]
    Therefore
    \begin{align*}
    -\log_M (\LL^{-1} \boldsymbol{1}_S)_{k_q} & = \frac{-\log_{nU} (\LL^{-1} \boldsymbol{1}_S)_{k_q}}{\log_{nU} (M)} 
    \geq \frac{-\log_{nU} (\abs{S} \cdot \max_{w' \in S}(\LL^{-1})_{k_q w'})}{\log_{nU} (M)}
    \\ & 
    \geq 
    \frac{- 1 + \min_{w' \in S} D_{\LL}(k_q, w')-2}{\log_{nU} (M)}
    \geq \frac{0.8 d_i - 7}{1.2} \geq 0.66 d_i - 6.
    \end{align*}
    Thus
    \[
    (\LL^{-1} \boldsymbol{1}_S)_{k_q} \leq M^{-0.66 d_i + 6}.
    \]
    Now since $U>32$ and $n>10$, we have $d_i \geq 32$, $M>320$, and 
    \[
    \yy_{k_q} \leq  M^{-0.66 d_i + 6} + M^{-2d_i} < M^{-d_i/2}.
    \]
    Therefore $k_q$ is not in $T$ and is not in the connected component of $v$. Therefore by this contradiction for all $w \in S \setminus \{v\}$, $C_v\cap C_w = \emptyset$.

    \item This holds because in each iteration of the loop (i), we pick at most one outer ball $W$ that contains a vertex $u$. Therefore each vertex is included in at most $6 \ell^2 \cdot 16^{\ell} \cdot \lceil \log(n/\delta)\rceil$ outer balls.

    \item For the last two parts, let $(V,W) \in \mathcal{C}$ be an inner ball, outer ball pair and $i,j$ be their corresponding iterand in the outer loop. Moreover let $v$ be the unique element in $S$ that is inside the connected component $C_v$ of $G_T$.
    Therefore
    \[
    V = \{k \in C_v: \yy_k \geq M^{-d_i/4} - M^{-2d_i}\}
    \]
    and 
    \[
    W= \{k \in C_v: \yy_k \geq M^{(-d_i/2) + 2}\}.
    \]
    Let $y \in W$. Then we have
    \[
    (\LL^{-1} \boldsymbol{1}_S)_y + M^{-2d_i} \geq \yy_y \geq M^{(-d_i/2) + 2}.
    \]
    Moreover
    \[
    (\LL^{-1} \boldsymbol{1}_S)_y = (\LL^{-1} \ee^{(v)})_y + \sum_{w \in S \setminus\{v\}} (\LL^{-1} \ee^{(w)})_y. 
     \]
     Next, we show, for $w \in S \setminus \{v\}$,
     \begin{align}
    \label{eq:outside-ball-less-m-di-2}
     (\LL^{-1} \ee^{(w)})_y < M^{-d_i/2}.
     \end{align}
    This is because for any $o \in [n] \setminus \{w\}$,
    \[
    \sum_{q \in [n]} \LL_{o q}(\LL^{-1}\ee^{(w)})_{q} = 0
    \]
    Therefore
    \[
    \sum_{q \in [n] \setminus \{o\}} \frac{-\LL_{o q}}{\LL_{oo}}(\LL^{-1}\ee^{(w)})_{q} = (\LL^{-1}\ee^{(w)})_{o}
    \]
    Since $\sum_{q \in [n] \setminus \{o\}} \frac{-\LL_{o q}}{\LL_{oo}} \leq 1$, either $(\LL^{-1}\ee^{(w)})_{o}$ is equal to $(\LL^{-1}\ee^{(w)})_{q}$ for all $q$ that are adjacent to $o$ or for one neighbor $q$, $(\LL^{-1}\ee^{(w)})_{q} > (\LL^{-1}\ee^{(w)})_{o}$. Continuing this argument inductively, one can see that if $(\LL^{-1}\ee^{(w)})_{y} \geq M^{-d_i/2}$, then there is a path from $y$ to $w$ where for all vertices $q$ on this path $(\LL^{-1}\ee^{(w)})_{q} \geq M^{-d_i/2}$. In this case $y \in C_w$ and therefore since $y$ is also in $C_v$, we have $C_v \cap C_w \neq \emptyset$ which is a contradiction.

     Therefore
     \[
     (\LL^{-1} \ee^{(v)})_y \geq M^{-(d_i/2)+2} - n M^{-d_i/2} -  M^{-2d_i} \geq M^{-(d_i/2)+1}.
     \]
     Thus
     \begin{align}
     \nonumber
     D_{\LL} (v,y)=-\log_{nU} ((\LL^{-1} \ee^{(v)})_y) + 2 & = -\log_M ((\LL^{-1} \ee^{(v)})_y) \log_{nU} M + 2 
     \\ \nonumber & \leq (\frac{d_i}{2} - 1)\log_{nU} M + 2 
     \\ & \leq 0.6 d_i -1.2 + 2\leq 0.65 d_i,
     \label{eq:inner-center-to-cluster-dist}
     \end{align}
     where the penultimate inequality follows from \eqref{eq:log-M-nU-bounds}.
     Now by triangle inequality (\Cref{claim:dist-prop}) for $y,y'\in W$,
     \[
     D_{\LL} (y,y') \leq D_{\LL} (v,y) + D_{\LL} (v,y') \leq 1.3 d_i.
     \]
     Therefore since the largest $d_i$ is $4^{\ell}$, the result follows.

    \item Let $y \in V$ and $z \in [n] \setminus W$. 
    Similar to the proof of the previous property, we can show that for all $w \in S \setminus\{v\}$,  $(\LL^{-1} \ee^{(w)})_y < M^{-d_i/2}$ (see \eqref{eq:outside-ball-less-m-di-2}). Moreover since $y \in V$, we have
    \[
    (\LL^{-1} \boldsymbol{1}_S)_y + M^{-2d_i} \geq \yy_y \geq M^{-d_i/4} - M^{-2d_i}.
    \]
    Therefore since $(\LL^{-1} \boldsymbol{1}_S)_y = (\LL^{-1} \ee^{(v)})_y + \sum_{w \in S \setminus\{v\}} (\LL^{-1} \ee^{(w)})_y$,
     \[
     (\LL^{-1} \ee^{(v)})_y \geq M^{-d_i/4} - n M^{-d_i/2} - 2 M^{-2d_i} \geq M^{-d_i/4+0.1}.
     \]
     Thus
     \begin{align}
     \nonumber
     D_{\LL} (v,y)=-\log_{nU} ((\LL^{-1} \ee^{(v)})_y) + 2 & = -\log_M ((\LL^{-1} \ee^{(v)})_y) \log_{nU} M + 2 
     \\ \nonumber & \leq (\frac{d_i}{4} - 0.1)\log_{nU} M + 2 
     \\ & \leq 0.3 d_i -0.08 + 2\leq 0.36 d_i,
     \label{eq:inner-center-to-cluster-dist}
     \end{align}
     where the penultimate inequality follows from \eqref{eq:log-M-nU-bounds}.
     Now by triangle inequality (\Cref{claim:dist-prop}) and \eqref{eq:inner-center-to-cluster-dist},
    \[
    D_{\LL}(v,z) \leq D_{\LL}(v,y) + D_{\LL}(y,z) \leq 0.36 d_i + D_{\LL}(y,z).
    \]
    Again, by an argument similar to the previous part for \eqref{eq:outside-ball-less-m-di-2}, we have $(\LL^{-1}\ee^{(v)})_{z} < M^{-d_i/2}$; otherwise, $z$ would belong to $C_v$ and thus to $W$, which is a contradiction. 
    Therefore by \eqref{eq:log-M-nU-bounds},
    \[
    D_{\LL}(v,z) = - \log_{nU} ((\LL^{-1}\ee^{(v)})_{z}) + 2 = \frac{-\log_M ((\LL^{-1}\ee^{(v)})_{z})}{\log_M (nU)} + 2 > \frac{d_i}{2} + 2,
    \]
    which gives
    \[
    D_{\LL} (y,z) \geq 0.14 d_i + 2.
    \]
    Thus since the smallest $d_i$ is $2^{\ell+1}$, the result follows.
\end{enumerate}

We now bound the bit complexity and failure probability of the algorithm.  
The algorithm performs a total of $2^{O(\sqrt{\log n})} \log(n) \log(n/\delta)$ iterations.  
In each iteration, we solve a linear system that takes
\[
\Otil(m \log (M^{4^{\ell}}) \log(U M^{4^{\ell}} \delta^{-1} \ell^2 B)) 
= m \, 2^{O(\sqrt{\log n})} \log(U) \log(U \delta^{-1})
\]
time.  
Moreover, in each iteration we compute connected components, which takes $O(m)$ time.  
Therefore, the total bit complexity of the algorithm is bounded by
\[
m \, 2^{O(\sqrt{\log n})} \log(U) \log(\delta^{-1}) \log(U \delta^{-1}).
\]

For the failure probability, note that since the total number of iterations is $\ell^2 B$, and each linear-system solve succeeds with probability at least $1 - \delta/(2\ell^2 B)$, by a union bound all linear systems are solved correctly with probability at least $1 - \delta/2$.  
Furthermore, as argued for the second property, we obtain a desirable set $S$ for each vertex with probability at least $1 - \delta/(2n)$ over the $B$ iterations of the inner loop for the corresponding $i$ and $j$.  
Hence, with probability at least $1 - \delta/2$, we select a desirable $S$ for all vertices throughout the algorithm, ensuring that all vertices are covered by the resulting low-diameter cover.  
By another union bound, the algorithm therefore succeeds with probability at least $1 - \delta$.

\end{proof}

\section{The Almost-Linear-Time Algorithm}
\label{sec:almost-lin-alg}

In this section, we present our algorithm for solving SDDM linear systems with entrywise approximation guarantees in almost-linear time and prove the main theorem of the paper (\Cref{thm:main}).  
To this end, we first recall the threshold decay (TD) framework from~\cite{GNY25}, its requirements, and the guarantees it provides in \Cref{sec:td-framework}.  
We then derive bounds on the error of solutions to partial linear systems solved within the TD framework in \Cref{sec:error-partial-solver}.  
Next, in \Cref{sec:partial-solver}, we show that our low-diameter cover can be used to construct the partial solvers required in the iterations of the TD framework.  
Finally, we provide implementation details for achieving almost-linear running time and complete the proof in \Cref{sec:final-almost-lin-time}.

\subsection{The Threshold Decay Framework}
\label{sec:td-framework}

Our almost-linear–time algorithm builds upon the TD framework proposed in~\cite{GNY25}.  
In this framework, the linear system is solved by iteratively solving a sequence of partial linear systems, where each matrix is a principal submatrix of the original SDDM matrix, and the right-hand side vector is adaptively updated based on previously computed entrywise approximations to certain entries---see \callalg{ThresholdDecay}.

\begin{figure}[!t]
    \begin{algbox}
        \textbf{\uline{\textsc{ThresholdDecay}}}: The algorithm framework for our solvers. Implementations are specified later. Only correctness is proved for the framework.  
        \\
        \uline{Input}: $\LL$: an $n \times n$ invertible RDDL matrix with integer entries in $[-U, U]$, \\
            $\bb$: a nonnegative vector with integer entries in $[0,U]$,
            \\
            $\epsilon$: the target accuracy, \\
            $T$: number of iterations (in this paper we always set it to $n$)\\
            Assume that $U \ge 2, T \ge 10, \eps \in (0, 1)$ \\
        \uline{Output:} vector $\xxtil$ supported on some set $A \subseteq [n]$ such that $\xxtil_A \approxbar_{\epsilon} (\LL^{-1} \bb)_A$.
        
        \begin{enumerate}
            \item Initialize $S^{(0)} = [n]$,
                $\bbhat^{(0)} = \bb$, and $\widetilde{\xx}^{(0)}$ to a zero dimensional vector. Let $\eps_L := \frac{\epsilon}{64T(nU)^2}$.
            \item For $t = 0 \ldots T$:
            \begin{enumerate}
                \item \label{ThresDec:L-solver}
                    Compute 
                    $\xxhat^{(t)}$ defined on $S^{(t)}$ such that
                    $ \norm{\widehat{\xx}^{(t)} - \LL_{S^{(t)}, S^{(t)}}^{-1} \widehat{\bb}^{(t)}}_2 \leq
                    \eps_L
                    \norm{\widehat{\bb}^{(t)}}_2$.
                \item \label{ThresDec:set-threshold}
                    Let $\theta_{t} $ be the smallest power of two larger than $  \frac{1}{4(nU)^2} \norm{\widehat{\bb}^{(t)}}_1$.
                \item \label{ThresDec:extractF}
                    Let $F^{(t)} \subseteq S^{(t)}$ be the indices of all the entries in $\xxhat^{(t)}$ larger than or equal to $\theta_{t}$, and let $S^{(t + 1)} \leftarrow S^{(t)} \setminus F^{(t)}$.
                \item \label{ThresDec:updateX}
                    Let $\xxtil^{(t+1)}$ be defined on $[n]\setminus S^{(t+1)}$ with $[\xxtil^{(t+1)}_{[n] \setminus S^{(t)}};\xxtil^{(t+1)}_{F^{(t)}}] \leftarrow [\xxtil^{(t)}; \xxhat^{(t)}_{F^{(t)}}]$.
                \item \label{ThresDec:updateB}
                    Create vector $\bbhat^{\left(t + 1\right)}$ for the next iteration such that
                    \begin{align} \label{eq:iteration-b}
                        \bbhat^{\left(t + 1\right)}
                        \approxbar_{\epsilon/(8T)}
                        \bb_{S^{\left( t +1 \right)}}
                        -
                        \LL_{S^{\left( t +1 \right)}, [n] \setminus S^{(t+1)}}
                        \xxtil^{\left( t+1 \right)}
                    \end{align}
            \end{enumerate} 
            
            \item Output set $A:=[n] \setminus S^{(T)}$ and $\xxtil := [\xxtil_A; \xxtil_{S^{(T)}}] = [\xxtil^{(T)}; \mathbf{0}]$
        \end{enumerate}
    \end{algbox}
    \caption{The Threshold Decay framework.}
    \labelalg{ThresholdDecay}
\end{figure}

\paragraph{Recall Threshold Decay.}
Let $\xxbar := \LL^{-1} \bb$ denote the solution vector. In Threshold Decay, $[n] = S^{(0)} \supset S^{(1)} \supset S^{(2)} \supset \dots$ are a chain of indices where $S^{(t)}$ is the set of indices for which an entrywise approximate has yet not been computed at iteration $t$. We invoke Laplacian solver on the set $S^{(t)}$ to solve $\xxhat^{(t)}$:
\begin{align} \label{eq:threshold-decay-solve}
    \| \xxhat^{(t)} - (\LL_{S^{(t)}, S^{(t)}})^{-1} \bbhat^{(t)} \|_2 \le \epsilon_L \| \bbhat^{(t)} \|_2
\end{align}
The framework has the following correctness guarantees.
\begin{lemma}[Lemma 3.7 of \cite{GNY25}] \label{lemma:threshold-decay-correctness}
Given $\eps \in (0, 1)$, integer $T \ge 10$, $U \ge 2$, an $n \times n$ invertible RDDL matrix $\LL$ with integer entries in $[-U, U]$, and a nonnegative vector $\bb$ with integer weights in $[0, U]$, let $\xxbar \defeq \LL^{-1} \bb$ denote the exact solution vector, \textsc{ThresholdDecay} has the following guarantees for all $t \in [0, T]$:
\begin{enumerate}
    \item $\norm{\bbhat^{(t)}}_1 \le (nU)^{-t} \norm{\bb}_1$. Moreover, if $t>0$, $\norm{\bbhat^{(t)}}_1 \le \frac{1}{(nU)} \norm{\bbhat^{(t-1)}}_1$.
    \item $\xxtil^{(t)} \approxbar_{\eps t / (4T)} \xxbar_{[n] \setminus S^{(t)}}$. Note that we define $\xxtil^{(0)}$ to be a zero dimensional vector.
    \item For all $i \in S^{(t+1)}$, $\xxbar_i < (nU)^{-2} \norm{\bbhat^{(t)}}_1 \le (nU)^{-(t+2)} \norm{\bb}_1$.
\end{enumerate}
\end{lemma}
The first property implies that the threshold $\|\bbhat^{(t)}\|_1$ decreases over iterations.  
The second property ensures that the computed values are accurate entrywise approximations.  
The third property guarantees that large entries are solved and removed from the set $S^{(t)}$ in iteration $t$.

We invoke the following theorem from~\cite{GNY25} to establish the correctness of the framework.  
Note that in this paper we apply the framework only to SDDM matrices, although it generalizes to RDDL matrices.

\begin{theorem} [Correctness of the Threshold Decay framework, Theorem 3.1 of \cite{GNY25}]\label{thm:ThresDec-correctness}
Given $\epsilon \in (0,1)$, integer $T\geq 10$, $U \ge 2$, an $n \times n$ invertible RDDL matrix $\LL$ with integer entries in $[-U, U]$, 
and a nonnegative vector $\bb$ with integer weights in $[0,U]$, \callalg{ThresholdDecay} produces a vector $\xxtil$ and a set $A \subseteq [n]$ such that $\xxtil$ has support $A$, with the following properties hold:  
    \begin{itemize}
        \item For any $i \in [n]$ with $(\LL^{-1} \bb)_i \geq (nU)^{-(T+1)} \norm{\bb}_1$, we have $i \in A$. 
        \item Our output is an entrywise $\exp(\eps)$-approximation on $A$, i.e.,  $\xxtil_A \approxbar_{\epsilon} (\LL^{-1} \bb)_A$.
        \item Specifically, when $T = n$, we have $\xxtil \approxbar_{\epsilon} \LL^{-1} \bb$.
    \end{itemize}
\end{theorem}

\subsection{Errors Corresponding to Partial Linear Systems}
\label{sec:error-partial-solver}

In this subsection, we present results that bound the error for solutions to partial linear systems, i.e., when a subset of entries is set to zero.  
We will use these bounds to establish the correctness of our algorithm.

\begin{lemma}
\label{lemma:remove-far-set}
Let $\bb$ be a nonnegative $n$-dimensional vector, and let $P$ denote the set of its positive entries.  
If $T \subseteq [n] \setminus P$ and $D_{\LL}(u, v) \ge d > 0$ for all $u \in P$ and $v \in T$, then
\[
    \bigl\|
        (\LL^{-1})_{-T, -T} \bb_{-T}
        -
        (\LL_{-T, -T})^{-1} \bb_{-T}
    \bigr\|_2
    < n^{2.5} U (nU)^{2-d} \norm{\bb}_1.
\]
\end{lemma}

\begin{proof}
Let
\[
    \LL =
    \begin{pmatrix}
        \LL_{-T, -T} & \LL_{-T, T}\\
        \LL_{T, -T} & \LL_{T, T}
    \end{pmatrix},
\]
and let $\SS := \LL_{T, T} - \LL_{T, -T} (\LL_{-T, -T})^{-1} \LL_{-T, T}$ be the Schur complement.  
Then, by \Cref{eq:block-Sc-inversion},
\begin{align*}
    (\LL^{-1})_{T, -T} &= - \SS^{-1} \LL_{T, -T} (\LL_{-T, -T})^{-1}, \\
    (\LL^{-1})_{-T, -T} &= (\LL_{-T, -T})^{-1}
    + (\LL_{-T, -T})^{-1} \LL_{-T, T} \SS^{-1} \LL_{T, -T} (\LL_{-T, -T})^{-1}.
\end{align*}
We can bound
\begin{align*}
    \bigl\|
        (\LL^{-1})_{-T, -T} \bb_{-T}
        -
        (\LL_{-T, -T})^{-1} \bb_{-T}
    \bigr\|_2
    & =
    \| (\LL_{-T, -T})^{-1} \LL_{-T, T} \SS^{-1} \LL_{T, -T} (\LL_{-T, -T})^{-1} \bb_{-T} \|_2
    \\ & = 
    \| (\LL_{-T, -T})^{-1} \LL_{-T, T} (\LL^{-1})_{T, -T} \bb_{-T} \|_2
    \\ & \le
    \| (\LL_{-T, -T})^{-1} \|_2
    \| \LL_{-T, T} \|_{\infty \to 2}
    \| (\LL^{-1})_{T, -T} \bb_{-T}\|_\infty
    \\ & \le
    n^2
    \cdot
    \sqrt{n} U
    \cdot
    (nU)^{2-d} \| \bb_{-T} \|_1
\end{align*}
where we used \Cref{lemma:norm2-inv} to bound $\| (\LL_{-T, -T})^{-1} \|_2$, the fact that the maximum entry of $\LL$ is $U$ to bound $\| \LL_{-T, T} \|_{\infty \to 2}$, and the property of the probability distance to bound $\| (\LL^{-1})_{T, -T} \|_{\max}$ in the last inequality.
\end{proof}

The previous lemma provides an upper bound on the error incurred when removing a set of vertices that are sufficiently far—in terms of probability distance—from the support of the right-hand side vector.  
We now extend this result to quantify the total error in approximating the full solution of the linear system by zeroing out the corresponding entries for these far vertices.  
Intuitively, the corollary shows that vertices far from the nonzero entries of $\bb$ contribute negligibly to the solution of $\LL \xx = \bb$, and their removal introduces only a small error.

\begin{corollary}
\label{coro:remove-far-vertices}
Let $\LL \in \Z^{n \times n}$ be an invertible SDDM matrix, and let $\bb$ be any nonnegative $n$-dimensional vector.  
Let $P$ denote the set of indices of its positive entries, and let $T \subseteq [n] \setminus P$.  
If 
$
    D_{\LL}(u, v) \ge d > 0 
$
for all $u \in P, v \in T$,
then
\[
    \Bigl\|
        \bigl[(\LL^{-1} \bb)_{-T};\, (\LL^{-1} \bb)_T\bigr]
        -
        \bigl[(\LL_{-T, -T})^{-1} \bb_{-T};\, \mathbf{0}\bigr]
    \Bigr\|_2
    \le (nU)^{-d+5} \norm{\bb}_2.
\]
\end{corollary}

\begin{proof}
Since $\bb_T = \0$, we have $(\LL^{-1})_{-T, -T} \bb_{-T} = (\LL^{-1} \bb)_{-T}$.  
Applying \Cref{lemma:remove-far-set} to $\bb_{-T}$, we obtain
\begin{align} 
\label{eq:thm_3_4-2}
    \bigl\|
        (\LL_{-T, -T})^{-1} \bb_{-T}
        -
        (\LL^{-1} \bb)_{-T}
    \bigr\|_2
    <
    (nU)^{4.5 - d} \norm{\bb_{-T}}_1
    \le
    (nU)^{4.5 - d} \sqrt{n} \norm{\bb_{-T}}_2.
\end{align}
Next, observe that
\[
    (\LL^{-1} \bb)_T
    =
    (\LL^{-1})_{T, P} \bb_P
    +
    (\LL^{-1})_{T, -P} \bb_{-P}
    =
    (\LL^{-1})_{T, P} \bb_P.
\]
Thus we have 
\begin{equation*}
    \norm{
        (\LL^{-1})_{T, P} \bb_P
    }_2
    \le
    nU \cdot  (nU)^{2-d} \norm{\bb_P}_1
    =
    (nU)^{3-d} \norm{\bb_P}_1.
\end{equation*}
\end{proof}

\subsection{The Efficient Partial Solver}
\label{sec:partial-solver}

In our \callalg{ThresholdDecay} framework, we introduce the boundary-expanded set for the use of our efficient partial solver in the next section. 

\paragraph{The boundary-expanded set.} Recall that each step of \callalg{ThresholdDecay}, we extract the large entries in the solution vector of the subsystem
\begin{align*}
    \LL_{S^{(t)}, S^{(t)}} \xx = \bbhat^{(t)}.
\end{align*}
Suppose that the right-hand-side vector $\bbhat^{(t)}$ has support $I \subseteq[n]$. The combinatorial interpretation of the solution is the sum of random walks on the subgraph $S^{(t)}$ starting from the vertices in $I$.
By leveraging the low-diameter cover, we enlarge the set 
$I$ to its \emph{boundary-expanded set}, ensuring that all but a small weighted fraction of the random walks remain inside this set.
This allows us to solve the subsystem restricted to the boundary-expanded set instead of 
$S^{(t)}$
, incurring only a small additive error, which is sufficient to accurately recover the large entries.
\begin{definition}[Boundary-expanded set]
\label{def:boundary-expanded-set}
Given a low-diameter cover $\calC = \{(V_i, W_i)\}_i$ and a subset $I \subseteq [n]$, we define the boundary-expanded set of $I$ with respect to $\calC$ as 
\begin{align*}
    \ExpSet_{\calC}(I) := \{ u \in [n]: (\exists (V_i, W_i) \in \calC) (u \in V_i \land  |W_i \cap I|>0)   \}.
\end{align*}     
\end{definition}

Note that $\ExpSet_{\calC}(I)$ is a superset of $I$ that contains all the inner balls whose outer ball touches $I$. An equivalent definition of $\ExpSet_{\calC}(I)$ is
\begin{align}  \label{eq:equiv-def-boundary-expanded-set}
    \ExpSet_{\calC}(I) := \bigcup_{(V_i, W_i) \in \calC:|W_i \cap I|> 0}  V_i.
\end{align}
We call standard Laplacian solvers on the boundary-expanded set $\ExpSet_{\calC}(I)$ so that the entries on $I$ have good approximations. 

\paragraph{The efficient partial solver.}

We present the main subroutine, \callalg{PartialSolve}, of our almost-linear-time algorithm that solves a partial system  $\LL_{S^{(t)}, S^{(t)}} \xx = \bbhat^{(t)}$ in the $t$-th iteration of \callalg{ThresholdDecay}.

\begin{figure}[!t]
    \begin{algbox}
        \textbf{\uline{\textsc{PartialSolve}}}: The subroutine that solves the system $\LL_{S, S} \xx = \bbhat$ using our low-diameter cover.
        \\
        \uline{Input}: Access to $\LL$ in \callalg{ThresholdDecay}, a low-diameter cover $\calC$ of $\LL$,
        \\
        a subset $S \subseteq [n]$, the right-hand-side vector $\bbhat$, 
        the accuracy parameter $\epsL$, and $\delta_L \in (0, 1)$.
        \\
        This procedure has access to the sets $I,H \subseteq[n]$.
        \\
        \uline{Input requirements}: We assume $I = \{ u \in S : \bbhat_u > 0\}$ and $H = \ExpSet_{\calC}( I) \cap S$.
        \\
        \uline{Output:} 
            vector $\xxhat$ such that 
            $\norm{
                {\xxhat}
                -
                (\LL_{S, S})^{-1} \, \bbhat
            }_2
            \le
            \epsL \norm{\bbhat}_2
            $.\\

            \begin{enumerate}
                \item
                    We call the Laplacian solver in \Cref{cor:near-linear-solver-SDDM} under fixed-point arithmetic,
                        \begin{align} \label{eq:SDDM-almost-linear-solve}
                            \xxhat_{H} \gets Z \cdot \textsc{LaplacianSolver}(\LL_{H, H}, \bbhat_{H} / Z, \epsL / 2, \delta_L).
                        \end{align}
                    where $Z \gets \norm{\bbhat_H}_1$ is the normalization factor.
                \item Output vector $\xxhat := [\xxhat_H; \xxhat_{S \setminus H}] = [\xxhat; \0] $
            \end{enumerate}
        
    \end{algbox}
    \caption{Solve the linear system restricted on a subset of vertices using the low-diameter cover}
    \labelalg{PartialSolve}
\end{figure}

We state our main lemma for the correctness of \callalg{PartialSolve}.

\begin{lemma}[Correctness of \callalg{PartialSolve}] \label{lemma:partial-solve-correctness}
    For any $\LL, \bbhat$, any $I, H$ that satisfy the input requirements, any $(r_1, r_2, s)$-cover $\calC$ such that $r_2 \ge 5 + \log_{n U} (2 / \epsL)$, the returned vector $\xxhat$ of \callalg{PartialSolve} satisfies
    \begin{align*}
        \norm{\xxhat - (\LL_{S,S})^{-1} \bbhat}_2 
        \le
        \epsL \norm{\bbhat}_2
    \end{align*}
    with probability $\ge 1 - \delta_L$.
\end{lemma}

\begin{proof}
    Recall that $\xxhat$ is obtained from \eqref{eq:SDDM-almost-linear-solve}, which satisfies
    \begin{align} \label{eq:lap-solver-guarantee}
        \norm{
            \xxhat_H
            -
            ( \LL_{H, H} )^{-1} \bbhat_H
        }_2
        \le
        \frac{\epsL}{2}
        \norm{
            \bbhat_H
        }_2.
    \end{align}
    with probability at least $1-\delta_L$ and $H = \ExpSet_{\calC} (I) \cap S$.

    Fix any $u \in S \setminus H$, we have $u \notin \ExpSet_{\calC} (I)$. By the definition of $\ExpSet_{\calC} (I)$, it holds that
    \begin{align*}
        (\forall i)(u \notin V_i \lor |W_i \cap I| = 0 ).
    \end{align*}
    By the definition of the clusters, there must exist an $i_u$ such that $u \in V_{i_u}$. Thus, $|W_{i_u} \cap I| = 0$, implying that $D_{\LL}(u,v) > r_2$ for all $v \in I$, due to the last property \Cref{def:low-diam-cover}. Note that the distance here is with respect to the Laplacian matrix of the entire graph. By \Cref{lem:distance-monotonicity}, we have $D_{\LL_{S, S}}(u,v) \ge D_{\LL}(u,v)$. Therefore, we have shown that 
    \begin{align}
        D_{\LL_{S, S}}(u,v) > r_2, \quad \forall u \in S \setminus H, v \in I.
    \end{align}
    
    We then apply \Cref{coro:remove-far-vertices} to the base Laplacian matrix $\LL_{S, S}$ and the right-hand-side vector $\bbhat$ defined on the set $S$, with $S \setminus H$ being the target set, obtaining 
    \begin{align}
        \label{eq:pseudo-D-removal-guarantee-1}
        \norm{
        \left( (\LL_{S, S})^{-1} \, \bbhat \right)_H
        -
        \left( \LL_{H, H} \right)^{-1} \, \bbhat_H 
        }_2
        &\le
        (nU)^{-r_2+5} \norm{\bbhat}_2\\
        \label{eq:pseudo-D-removal-guarantee-2}
        \norm{
        \left( (\LL_{S,S})^{-1} \bbhat \right)_{S \setminus H}
        }_2
        &\le
        (nU)^{-r_2+5} \norm{\bbhat}_2
    \end{align}
    
    Setting $r_2 \ge 5 + \log_{nU} (2 / \epsL)$, triangle inequality with \Cref{eq:lap-solver-guarantee} and \Cref{eq:pseudo-D-removal-guarantee-1} achieves
    \begin{align}
        \norm{
            \xxhat_H
            -
            \left( (\LL_{S,S})^{-1} \bbhat \right)_H
        }_2
        \le
        \left( \frac{\epsL}{2} + (nU)^{-r_2 + 5} \right)
        \norm{\bbhat}_2 
        \le \label{eq:correctness-1}
        \epsL \norm{\bbhat}_2.
    \end{align}

    Recall that $\xxhat_{S \setminus H} = \0$. Thus, \Cref{eq:pseudo-D-removal-guarantee-2} achieves
    \begin{align}
        \norm{
        \xxhat_{S \setminus H}
        -
        \left( (\LL_{S,S})^{-1} \bbhat \right)_{S \setminus H}
        }_2
        \le
        (nU)^{-r_2+5} \norm{\bbhat}_2
        < \label{eq:correctness-2}
        \epsL \norm{\bbhat}_2.
    \end{align}

    The correctness lemma then follows from combining \Cref{eq:correctness-1} and \Cref{eq:correctness-2}.
\end{proof}

\subsection{The Almost-Linear-Time SDDM Solver}
\label{sec:final-almost-lin-time}

\begin{figure}[!t]
    \begin{algbox}
        \textbf{\uline{\textsc{SDDMSolve}}}: The entrywise approximate linear system solver for SDDM matrices based on \callalg{ThresholdDecay} and \callalg{PartialSolve}, with fast implementations. \\
        \uline{Input}: 
            $\LL, \bb, \eps, \kappa$ of \callalg{ThresholdDecay}, where $\LL$ is an invertible SDDM matrix. $\delta \in (0,1)$.\\
            WLOG, assume $\LL$ is irreducible. \\
        \uline{Output:} vector $\xxtil$ such that $\xxtil \approxbar_{\epsilon} \LL^{-1} \bb$. \\
        \begin{enumerate}
            \item  \label{line:construct-calC}
                Produce a low-diameter cover $\calC \gets \callalg{LowDiamConstruct}(\LL, \delta / 2)$.
            \item  \label{line:run-threshold-decay}
                Run \callalg{ThresholdDecay}, with vectors $\bbhat^{(t)}, \xxtil^{(t)}, \xxhat^{(t)}$ represented by $O(\log(nU/\eps))$-bit floating-point numbers. \\
                \textbf{Step \ref{ThresDec:L-solver}: }
                We maintain $I^{(t)} := \{u \in S^{(t)} : \bbhat^{(t)}_u > 0\}$ and $H^{(t)} := \ExpSet_{\calC}(I^{(t)}) \cap S^{(t)}$ efficiently over the iterations. In iteration $t$:
                \begin{enumerate}
                    \item $\xxhat^{(t)} \gets \callalg{PartialSolve}(\LL, \calC, S^{(t)}, \bbhat^{(t)}, \epsL, \delta / (2 n), H^{(t)}, I^{(t)})$.
                \end{enumerate}
                \textbf{Step \ref{ThresDec:set-threshold}, \ref{ThresDec:extractF}, \ref{ThresDec:updateX}: } Work only on the subset of indices $H^{(t)}$; ignore the entries on $S^{(t)}$.

                \textbf{Step \ref{ThresDec:updateB}: } Initialize $n$-dimensional vector $\vvhat^{(0)} = 0$. In iteration $t$:
                \begin{enumerate}
                    \item
                        Compute $\vvhat^{(t+1)}$ defined on $S^{(t+1)}$, using in-place updates to $\vvhat^{(t)}$, 
                        \begin{align} \label{eq:compute-v} 
                            \vvhat^{(t+1)} \gets \vvhat^{(t)}_{S^{(t+1)}} 
                            -
                            \LL_{S^{\left( t +1 \right)}, F^{(t)}}
                            \xxhat^{\left( t \right)}_{F^{(t)}}.
                        \end{align}
                    \item 
                        Then, we have an implicit representation for $\bbhat^{(t+1)}$, that is, 
                        \begin{align} \label{eq:bhat-rep} 
                            \bbhat^{\left(t + 1\right)}
                            :=
                            \bb_{S^{\left( t +1 \right)}}
                            +
                            \vvhat^{(t+1)}.
                        \end{align}
                \end{enumerate}
        \end{enumerate}
        Return the vector $\xxtil$ produced by \callalg{ThresholdDecay}.
    \end{algbox}
    \caption{The entrywise approximate linear system solver for SDDM matrices in almost-linear-time}
    \labelalg{SDDMSolve}
\end{figure}

In this section, we present our almost-linear-time SDDM solver, \callalg{SDDMSolve}, by plugging in \callalg{PartialSolve} into \callalg{ThresholdDecay} with our pre-computed low-diameter cover from \callalg{LowDiamConstruct}.
The correctness of our solver follows directly from the correctness of the subroutines (\Cref{lemma:threshold-decay-correctness}, \Cref{lemma:partial-solve-correctness}, and \Cref{thm:construct-low-diam-cover}). We focus on fast implementations and the resulting almost-linear running time bounds. 

In iteration $t$, we call \callalg{PartialSolve} on the subsystem $\LL_{S^{(t)}, S^{(t)}}$ defined on the subset $S^{(t)}$. 
We use $I^{(t)}, H^{(t)}, \eta^{(t)}, \xxhat^{(t)}$ to denote $I, H, \eta, \xxhat$ of \callalg{PartialSolve} invoked in the $t$-th iteration. In the worst case, the sum of the sizes of $S^{(t)}$ over all iterations is $\Theta(n^2)$. Our partial solver works on the smaller subset $H^{(t)}$ instead of $S^{(t)}$. The following lemma shows that each vertex occurs in only a small number of sets $H^{(t)}$'s. 

\begin{lemma} \label{lemma:H-small}
    In \callalg{SDDMSolve}, for any $u \in [n]$, it holds that $\sum_{t=0}^{T} \mathbf{1}[u \in H^{(t)}] \le  \rin$. 
\end{lemma}

Note that by \Cref{thm:construct-low-diam-cover}, we have that $\rin =  2^{O(\sqrt{\log n})}$ is small.
We first prove the lemma in the next section. To get the almost-linear running time, we also discuss the efficient maintenance of the boundary-expanded set $\ExpSet_{\calC}(S^{(t)})$ in \callalg{PartialSolve}, as well as fast implementations of the updates to the vectors and sets in \callalg{ThresholdDecay}.

\subsubsection{Bounding the sum of sizes}

In this section, we bound the sum of sizes of $H^{(t)}$ over all iterations $t$. We need the following lemma from \cite{GNY25}.

\begin{lemma}[Lemma 3.14 of \cite{GNY25}]
    \label{lemma:solution-of-neighbors-single-system}
    Let $\LL$ be an $n\times n$ invertible SDDM matrix with integer weights in $[-U,U]$, and $\bb$ be a nonnegative vector of dimension $n$.
    For any $i,j\in[n]$ with $\LL_{ij} \neq 0$, we have
    \begin{align*}
        (nU)^{-1} \cdot (\LL^{-1} \bb)_j \le (\LL^{-1} \bb)_i \leq nU \cdot (\LL^{-1} \bb)_j.
    \end{align*}
\end{lemma}

We also need the following generalized version of the lemma, using our probability distance (defined in \Cref{def:probability-distance}).

\begin{lemma} \label{lemma:solution-of-near-vertices}
    Let $\LL$ be an $n \times n$ invertible SDDM matrix with integer weights in $[-U, U]$, and $\bb$ be a nonnegative vector of dimension $n$.
    For any $i,j\in[n]$ with $D_{\LL}(i, j) \le d$, we have
    \begin{align*}
        (nU)^{-d} \cdot (\LL^{-1} \bb)_j \le (\LL^{-1} \bb)_i \leq (nU)^{d} \cdot (\LL^{-1} \bb)_j
    \end{align*}
\end{lemma}
\begin{proof}
    Fix any $k \in [n]$. By the triangle inequality from \Cref{claim:dist-prop}, we have $D_{\LL}(i,k) \le D_{\LL}(i,j) + D_{\LL}(j,k)$, and thus 
    \begin{align} \label{eq:tmp-Linverse-bound}
        (\LL^{-1})_{ik} \ge (\LL^{-1})_{ij} (\LL^{-1})_{jk}  (nU)^{-2} 
        \ge
        (nU)^{-d+2} \cdot (\LL^{-1})_{jk}  (nU)^{-2}
        =
        (nU)^{-d} \cdot (\LL^{-1})_{jk}.
    \end{align}
    Note that $(\LL^{-1} \bb)_i$ and $(\LL^{-1} \bb)_j$ can be written as linear combinations:
    \begin{align*}
        (\LL^{-1} \bb)_i = \sum_{k=1}^{n} (\LL^{-1})_{ik} \bb_k,
        \quad
        (\LL^{-1} \bb)_j = \sum_{k=1}^{n} (\LL^{-1})_{jk} \bb_k.
    \end{align*}
    Since $\bb_k \ge 0$ for all $k$, the inequality
    $
    (\LL^{-1} \bb)_i \ge (nU)^{-d} \cdot (\LL^{-1} \bb)_j
    $
    follows by taking the same linear combination of both sides of \eqref{eq:tmp-Linverse-bound}. The other direction of the lemma follows by symmetry.
\end{proof}

\begin{proof}[Proof of \Cref{lemma:H-small}.]
We fix any $u \in [n]$. Suppose that in iteration $t$, $u$ is contained in the set
\begin{align*}
    H^{(t)} = \ExpSet_{\calC}(I^{(t)}) \cap S^{(t)}.
\end{align*}
We thus have $u \in S^{(t)}$ and $u \in \ExpSet_{\calC}(I^{(t)})$. By the definition of boundary-expanded set, there exists $(V, W) \in \calC$ such that $u \in V$ and $|W \cap I^{(t)}|>0$. Let $v$ be any vertex in $W \cap I^{(t)}$. By the definition of $I^{(t)}$, we have
$
    \bbhat^{(t)}_v>0.
$
Let $\xxbar := \LL^{-1}\bb$. We now lower bound $\xxbar_{v}$.

By step~\ref{ThresDec:updateB} of \callalg{ThresholdDecay}, we have
\begin{align*} 
    \bbhat^{\left(t \right)}
    \approxbar_{\epsilon/(8T)}
    \bb_{S^{\left( t  \right)}}
    -
    \LL_{S^{\left( t  \right)}, [n] \setminus S^{(t)}}
    \xxtil^{\left( t \right)}.
\end{align*}
Therefore, either $\bb_{v} > 0$, or there exists $w \in [n] \setminus S^{(t)}$ such that $|\LL_{v, w}| > 0$ and $\xxtil^{(t)}_w > 0$.
In the first case, $\xxbar_v \ge \bb_v \ge 1$. In the second case, by \Cref{lemma:solution-of-neighbors-single-system}, 
\begin{align*}
    \xxbar_v \ge (nU)^{-1} \xxbar_{w} \approxbar_{\eps t / (4T)} \xxtil_{w} \ge \theta_{t-1} = \frac{1}{4(nU)^2} \norm{ \bbhat^{(t-1)} }_1.
\end{align*}
where the entrywise approximation follows from the second property of \Cref{lemma:threshold-decay-correctness}, and the last inequality follows from step~\ref{ThresDec:extractF} of \callalg{ThresholdDecay}. We thus obtain the following lower bound (where we define $\bbhat^{(-1)}$ as the all-ones vector to handle the first case)
\begin{align*}
    \xxbar_v \ge
    \frac{1}{8(nU)^2} \norm{ \bbhat^{(t-1)} }_1.
\end{align*}
Since $u, v$ are in the same outer ball $W$ in the low-diameter cover $\calC$, by the fourth property of \Cref{def:low-diam-cover}, we have $D_{\LL}(u, v) \le \rin$.
Therefore, by \Cref{lemma:solution-of-near-vertices}, 
\begin{align} \label{eq:xu-lower-bound}
    \xxbar_u \ge \xxbar_v \cdot (nU)^{-\rin} 
    \ge
    (nU)^{-\rin-3} \norm{ \bbhat^{(t-1)} }_1.
\end{align}
If $u$ survives after $\rin + 1$ iterations, then $u \in S^{(t + \rin + 1)}$ and by \Cref{lemma:threshold-decay-correctness},
\begin{equation*}
    \xxbar_u < (nU)^{-2} \norm{\bbhat^{(t+\rin)}}_1 < (nU)^{-3 -\rin}  \norm{\bbhat^{(t-1)}}_1
\end{equation*}
where we repeatedly applied property 1 of \Cref{lemma:threshold-decay-correctness} for the second inequality.
This contradicts with \eqref{eq:xu-lower-bound}.
Therefore, for any $u \in [n]$, it must survive at most $\rin$ iterations, and thus
\begin{align*}
    \sum_{t=0}^{T} \mathbf{1}[u \in H^{(t)}] \le  \rin.
\end{align*}

\end{proof}

\subsubsection{Maintaining the vector \texorpdfstring{$\bbhat^{(t)}$}{bbhat^(t)}}

The implementation of step~\ref{ThresDec:updateB} (stated in \callalg{SDDMSolve}) is the same implementation as in \textsc{PredictiveSDDMSolver} of \cite{GNY25}.
This maintains the vector $\bbhat^{(t)}$ efficiently over the iterations by storing an extra vector $\vvhat^{(t)}$.

Note that $F^{(t)} = S^{(t)} \setminus S^{(t+1)}$ is the set of removed vertices in iteration $t$. We update $\bbhat^{(t)}$ on the edges from $F^{(t)}$ to $S^{(t+1)}$, i.e., the nonzero entries of $\LL_{S^{(t+1)}, F^{(t)}}$. Each edge only contributes once over all iterations, so the total number of arithmetic operation is $O(n+m)$. The detailed proof appears in the proof of Lemma 3.13 of \cite{GNY25}.

\begin{claim}[Lemma 3.13 of \cite{GNY25}] \label{claim:maintain-bbhat}
    The implementation of Step~\ref{ThresDec:updateB} in \callalg{SDDMSolve} runs in a total of $O((n+m) \log(nU/\eps))$ bit operations. There are only $O(n+m)$ total entry updates to the vector $\vvhat^{(t)}$ and $\bbhat^{(t)}$.
\end{claim}

 Furthermore, we need to track the norms of the vector $\bbhat^{(t)}$, including $\norm{\bbhat^{(t)}}_1$ and $\norm{\bbhat^{(t)}}_2$, for use in other steps of the algorithm. This can be implemented efficiently as well.

\begin{claim} \label{claim:maintain-bbhat-norm}
    In \callalg{SDDMSolve}, for any constant $c>0$, we can maintain $\norm{\bbhat^{(t)}}_1$ and $\norm{\bbhat^{(t)}}_2$ up to multiplicative error of $(\eps/(nU))^c$ with total $O((n+m) \log (nU/\eps) \log n)$ bit operations.
\end{claim}
\begin{proof}
     By \Cref{claim:maintain-bbhat}, there are only $O(n+m)$ total updates to the vector $\vvhat^{(t)}$. Observe that it suffices to compute the sum of $p$-th powers over the index set $S^{(t)}$, since 
     \begin{align*}
         \norm{\bbhat^{(t)}}_p^p = \sum_{i \in S^{(t)} } (\bb_i + \vv^{(t)}_i)^p.
     \end{align*}
     We can maintain this sum while updating $S^{(t)}$ and $\vv^{(t)}$. There are $n$ vertex removals from $S^{(t)}$ and $O(n+m)$ entry updates to $\vv^{(t)}$. However, directly subtracting the sum when removing vertices from $S^{(t)}$ can cause numerical issues due to cancelation. To avoid this, we use dynamic data structures that maintain partial sums. A segment tree supports $O(\log n)$ update and query time per operation. We conclude the claim by using an $O(\log(nU/\eps))$-bit floating-point number for each node of the segment tree.
\end{proof}

\subsubsection{Maintaining $I^{(t)}$ and $H^{(t)}$}

In this section, we discuss how the set $I^{(t)}$ and $H^{(t)}$ are maintained efficiently. Recall the input requirements for \callalg{PartialSolve}, $I^{(t)} = \{ u \in S^{(t)} : \bbhat^{(t)}_u > 0\}$, and $H^{(t)} = \ExpSet_{\calC}(I^{(t)}) \cap S^{(t)}$.

\paragraph{Maintaining $I^{(t)}$.}
Note that there is no numerical issue here. Since $\bbhat^{(t)}$ is constructed with multiplicative approximation guarantees (and represented using floating points), it is easy to distinguish a nonzero value from a zero, as no cancelation occurs throughout our algorithm. 

By \Cref{claim:maintain-bbhat}, there are only $O(n+m)$ total entry updates to $\bbhat^{(t)}$. Consequently, maintaining $I^{(t)}$ alongside $\bbhat^{(t)}$ is straightforward in total time $O(n+m)$.

$I^{(t)}$ has a more desirable property, as shown in the next claim.
\begin{claim}
    For any iteration $t < T$ and any vertex $u \in S^{(t+1)}$, if $\bbhat^{(t)}_u > 0$, then it follows that $\bbhat^{(t+1)}_u > 0$.
\end{claim}
\begin{proof}
    It follows directly from the construction of $\bbhat^{(t)}$, \eqref{eq:compute-v} and \eqref{eq:bhat-rep}, and the fact that $\LL_{S^{(t+1)}, F^{(t)}}$ only has non-positive entries while $\xxhat^{(t)}_{F^{(t)}}$ has positive entries.
\end{proof}
The claim shows that each vertex $u$ is inserted to $I^{(t)}$ only once.
It is removed from $I^{(t)}$ only when $u$ is no longer in the set $S^{(t)}$.

\paragraph{Maintaining $H^{(t)}$.} By the equivalent definition of boundary-expanded set \eqref{eq:equiv-def-boundary-expanded-set}, 
\begin{align*}
    \ExpSet_{\calC}(I) 
    & =
    \bigcup_{(V_i, W_i) \in \calC : |W_i \cap I|> 0}  V_i
     =
    \bigcup_{u \in I} \biggl(
        \bigcup_{
                (V_i, W_i) \in \calC:u \in W_i
        }  V_i
    \biggr)
    =
    \bigcup_{u \in I} \Cores_{\calC}(u).
\end{align*}
where $
\Cores_{\calC}(u):= \bigcup_{
                (V_i, W_i) \in \calC:u \in W_i
        }  V_i
$
is the union of cores of the balls containing $u$. We maintain the set $H^{(t)} = \ExpSet_{\calC}(I^{(t)}) \cap S^{(t)}$ as follows:
\begin{itemize}
    \item Create a counter $\cnt_i$ for each cluster $i \in |\calC|$, with initial value $0$. Maintain the indices of positive counters $C_{>0} := \{ i : \cnt_i > 0 \}$. 
    \item Insert/Delete $u$ in $I^{(t)}$: For each $i$ such that $(V_i, W_i) \in \calC$ and $u \in W_i$, increment/decrement $\cnt_i$ by $1$ and update $C_{>0}$.
    \item Build $H^{(t)}$ in iteration $t$: Build the union set by removing duplicates
    \begin{align*}
        H^{(t)} \gets \bigcup_{i \in C_{>0}} V_i.
    \end{align*}
\end{itemize}
By the third property of \Cref{def:low-diam-cover} and that $\calC$ is an $(\rin, \rout, \alpha)$-cover, each vertex $u$ is contained in at most $\alpha$ outer balls in $\calC$.

For the second step, since each vertex is inserted to $I^{(t)}$ only once, the total number of updates to the counters is bounded by $2n \alpha = n 2^{O(\sqrt{\log n})}$.

For the third step, $H^{(t)}$ is constructed in time complexity $O(|H^{(t)}| \cdot \alpha)$, since each vertex $u$ in $H^{(t)}$ has at most $\alpha-1$ duplicates by the third property of $\calC$.

The total time complexity for maintaining $H^{(t)}$ throughout the iterations is
\begin{align*}
    O(n 2^{O(\sqrt{\log n})}) + \sum_{t=0}^{T} O(|H^{(t)}| \cdot \alpha)
    & =
    O(n 2^{O(\sqrt{\log n})}) + \alpha  \cdot O\left( \sum_{t=0}^{T} |H^{(t)}|\right) 
    \\ & =
    O(n 2^{O(\sqrt{\log n})}) + \alpha\cdot  n 2^{O(\sqrt{\log n})} 
    =
    n 2^{O(\sqrt{\log n})}.
\end{align*}
where the second equality follows from \Cref{lemma:H-small}. We summarize our discussions into the following claim.

\begin{claim} \label{claim:maintain-H}
    In \callalg{SDDMSolve}, $H^{(t)}$ can be maintained in time complexity $n 2^{O(\sqrt{\log n})}$ throughout the iterations.
\end{claim}

\subsubsection{Wrap-up}

In this section, we put everything together to prove the main theorem of the paper.

\MainTheorem*

\begin{proof}
We start by recalling the parameters.
\paragraph{Parameters.} In Step~\ref{line:construct-calC} of \callalg{SDDMSolve}, it follows from \Cref{thm:construct-low-diam-cover} that $\calC$ is an $(\rin, \rout, \alpha)$-cover with parameters as below, for $\ell := \lceil \sqrt{\log n} \rceil + 3$:
\begin{itemize}
    \item Diameter: $\rin = 2^{2\ell+1} = 2^{O(\sqrt{\log n})}$.
    \item Outer gap: $\rout = 2^{\ell-2} = 2^{\Omega(\sqrt{\log n})}$.
    \item Size parameter: $\alpha = 6 \ell^2 \cdot 16^\ell \cdot \lceil \log(n/\delta) \rceil =  2^{O(\sqrt{\log n})}$.
\end{itemize}
This cover is produced in $\Otil( m n^{o(1)} \log(U \delta^{-1}) )$ bit operations.

In Step~\ref{line:run-threshold-decay} of \callalg{SDDMSolve}, we run \callalg{ThresholdDecay} for $T=n$ iterations.

\paragraph{Correctness.} We apply \Cref{lemma:partial-solve-correctness} to show that \callalg{PartialSolve} satisfies the guarantees of Step~\ref{ThresDec:L-solver} in \callalg{ThresholdDecay}. It suffices to verify 
\begin{align*}
    5 + \log_{nU}(2 / \epsL) = 5 + \log_{nU}(128 T (nU)^2 / \eps) \le 9 + \log_{nU}(1 / \eps) \le 9 + 2^{\sqrt{\log n}}< \rin
\end{align*}

Since $\rin = 2^{2\ell + 1}$, the inequality holds provided that $\eps \ge 1/\poly(n)$. It then follows from the correctness of \callalg{ThresholdDecay} (\Cref{thm:ThresDec-correctness}) that $\xxtil \approxbar_\eps \LL^{-1} \bb$, concluding the correctness of our almost-linear-time solver.

About the probability correctness:
\begin{itemize}
    \item Failure probability by using \callalg{LowDiamConstruct} is $\le \delta / 2$. We call this procedure once.
    \item Failure probability by using \callalg{PartialSolve} is $\le \delta / (2 n)$. We call this procedure once each iteration, so overall $n$ times.
\end{itemize}

Union bound gives the failure probability of \callalg{SDDMSolve} to be $\le \delta$, thus the algorithm succeeds with probability at least $1 - \delta$.

\paragraph{Bit complexity.}
By \Cref{lemma:H-small}, for any $u \in [n]$, we have
\begin{align*}
    \sum_{t=0}^{T} \mathbf{1}[u \in H^{(t)}] \le  \rin = 2^{O(\sqrt{\log n})},
\end{align*}
and consequently, the total size of $H^{(t)}$ over all iterations is almost-linear in $n$, i.e., 
\begin{align} \label{eq:sum-of-H-size}
    \sum_{t=0}^T |H^{(t)}| \le n 2^{O(\sqrt{\log n})}.
\end{align}
By \Cref{claim:maintain-H}, the set $H^{(t)}$ is maintained in $n 2^{O(\sqrt{\log n})}$ operations
throughout the iterations.

\textbf{Step 2a} is the dominating step. We bound the total time of the Laplacian solvers we invoked. By \Cref{cor:near-linear-solver-SDDM}, the Laplacian solver in iteration $t$ solves the subsystem $\LL_{H^{(t)}, H^{(t)}}$, and requires $\Otil( \nnz(\LL_{H^{(t)}, H^{(t)}}) \log^2(U \eps^{-1} \delta^{-1}))$ bit operations. The sum of sparsity can be bounded by
\begin{align*}
    \sum_{t=0}^T  \nnz(\LL_{H^{(t)}, H^{(t)}}) 
    & \le  \sum_{t=0}^{T} \sum_{u \in [n]} \mathbf{1}[u \in H^{(t)}] \cdot \deg(u)
    \\ & \le  \sum_{u \in [n]} \left(\sum_{t=0}^{T}  \mathbf{1}[u \in H^{(t)}]\right) \cdot \deg(u) 
    \le  m2^{O(\sqrt{\log n})}.
\end{align*}
For \textbf{Steps 2b, 2c, 2d}, every vector is supported only on the set $H^{(t)}$. A straightforward implementation that iterates over $H^{(t)}$ gives a total of $\Otil(2^{O(\sqrt{\log n})} \cdot n)$ bit operations. For example, in \textbf{Step 2c}, we need to compute the set $F^{(t)} \subseteq S^{(t)}$ such that
\begin{align*}
    F^{(t)} = \{ i \in S^{(t)}: \xxhat^{(t)}_i \ge \theta_t\}.
\end{align*}
We cannot afford enumerating over $S^{(t)}$; instead, we enumerate over $H^{(t)}$ and skip the indices in $S^{(t)} \setminus H^{(t)}$ since we know they are all zeros. By \Cref{eq:sum-of-H-size}, this step take a total of almost-linear time. 

Lastly, by \Cref{claim:maintain-bbhat} and \Cref{claim:maintain-bbhat-norm}, \textbf{Step 2e} is implemented in $O((n+m) \log(nU/\eps))$ time.
Putting everything together, we conclude that \callalg{SDDMSolve} runs in 
$\Otil(m2^{O(\sqrt{\log n})} \log^{2}(U\eps^{-1} \delta^{-1}))$ bit operations. The final running time follows by considering the running time of computing the low-diameter cover as well.
\end{proof}

\section{Conclusion}
\label{sec:conclusion}

We presented an almost-linear–time algorithm for solving SDDM linear systems with entrywise approximation guarantees. A natural open question is whether one can achieve truly near-linear running time, that is, $O(m \poly\log n)$. In this work, we focused on inputs given in fixed-point representation. However, as noted in \cite{GNY25}, the case of floating-point inputs presents additional challenges due to possible reductions from shortest-path problems. Investigating algorithms for SDDM systems with floating-point inputs therefore remains an interesting direction for future research.

\bibliographystyle{alpha}
\bibliography{main}

\end{document}